\documentclass{article}[12pt]
 
\usepackage[letterpaper, total={6in, 8in}]{geometry}

\usepackage{graphicx}
\usepackage{amssymb,amsmath,amsthm}
\usepackage{bbm}
\usepackage{thmtools}
\usepackage{titlesec}
\usepackage[svgnames]{xcolor}
\definecolor{blu}{RGB}{0,114,188}
\usepackage[colorlinks=True,allcolors=blu]{hyperref}
\usepackage{orcidlink}
\usepackage[compat=1.1.0]{tikz-feynman}
\tikzfeynmanset{warn luatex=false}
\usepackage{tikz}
\usetikzlibrary{calc}
\usepackage{float}
\usepackage{simpler-wick}
\usepackage[sort,compress,numbers,square]{natbib}
\usetikzlibrary{cd}
\usepackage[font=small,width=1\linewidth]{caption}
\usepackage{complexity}
\usepackage[absolute,overlay]{textpos}
\usepackage{longtable}

\declaretheorem[name=Definition]{definition}
\declaretheorem[name=Theorem, numberlike=definition]{theorem}
\declaretheorem[name=Proposition, numberlike=definition]{prop}
\declaretheorem[name=Lemma, numberlike=definition]{lem}
\declaretheorem[name=Corollary,numberlike=definition]{cor}
\usepackage{cleveref}
\crefname{thm}{Theorem}{Theorems}
\crefname{defn}{Definition}{Definitions}
\crefname{prop}{Proposition}{Propositions}
\crefname{lem}{Lemma}{Lemmas}
\crefname{cor}{Corollary}{Corollaries}

\numberwithin{definition}{section}
\numberwithin{prop}{section}
\numberwithin{theorem}{section}
\numberwithin{cor}{section}

\titleformat{\section}{\large\bfseries}{\thesection.}{1em}{}

\title{Dirac Traces and the Tutte Polynomial}
\author{Joshua Lin \orcidlink{0000-0002-3353-1559} \\ \small{joshlin@mit.edu}}
\date{\small{\textit{Center for Theoretical Physics, Massachusetts Institute of Technology, \\Cambridge, MA 02139, USA}\\{Preprint Number: MIT-CTP/5793}}}

\begin{document}

\renewcommand{\abstractname}{\vspace{-\baselineskip}}
\maketitle
\vspace{-1cm}
\begin{abstract}
\noindent Perturbative calculations involving fermion loops in quantum field theories require tracing over Dirac matrices. A simple way to regulate the divergences that generically appear in these calculations is dimensional regularisation, which has the consequence of replacing $4$-dimensional Dirac matrices with $d$-dimensional counterparts for arbitrary complex values of $d$. In this work, a connection between traces of $d$-dimensional Dirac matrices and computations of the Tutte polynomial of associated graphs is proven. 
The time complexity of computing Dirac traces is analysed by this connection, and improvements to algorithms for computing Dirac traces are proposed. 
\end{abstract}

\section{Introduction}

Quantum Field Theory (QFT) is a well-established theoretical framework describing various physical systems, from particle physics to condensed matter systems. Often the QFTs of interest are not exactly solvable, and only perturbative expansions in the coupling strength are analytically calculable. A general phenomena is the divergence of loop integrals appearing in perturbative calculations, requiring a choice of a regulator for the theory to render intermediate quantities finite. A particularly convenient regulator often used is dimensional regularisation \cite{tHooft:1972tcz}, which analytically continues the dimension of spacetime to an arbitrary complex parameter $d$. For QFTs containing fermions, fermion loops come with dimensionally regulated traces over associated products of Dirac matrices. For example, in Quantum Electro-Dynamics (QED) with a single massless fermion, the two-loop photon vacuum polarization (contracted against the metric tensor for simplicity) contains the following diagram as one of the contributions:
\begin{equation}\label{eq:examplecalc}
\begin{tikzpicture}[baseline=(a.base)]
    \begin{feynman}
      \vertex (a) at (0,0);
      \vertex (b) at (1,0);
      \vertex (c) at (3,0);
      \vertex (d) at (4,0);
      
      \vertex (g) at (2,0.95);
      \vertex (j) at (2,-0.95);
      \diagram*{
        (a) -- [gluon] (b);
        (c) -- [gluon,momentum=$k$] (d);
        (b) -- [fermion, quarter left,momentum=$k_1$] (g) -- [fermion, quarter left] (c) -- [fermion,quarter left] (j) -- [fermion, quarter left] (b);
        (g) -- [gluon,momentum=$k_2$] (j);
      };
      \node at (-0.25,-0.05) {$\nu$};
      \node at (4.25,-0.05) {$\nu$};
      \node at (7.6,0.5) {$=e^4 \ \mathrm{Tr}_d\left(\gamma_{\nu} \gamma_{\mu_1} \gamma_{\mu_5} \gamma_{\mu_2} \gamma_{\nu} \gamma_{\mu_3} \gamma_{\mu_5} \gamma_{\mu_4} \right)$};
      \node at (9.5,-0.5) {$\times {\displaystyle{\int}} \dfrac{\mathrm{d}^d k_1 \mathrm{d}^d k_2}{(2\pi)^{2d}} \dfrac{(k_1 - k)^{\mu_1} (k_1 - k_2 - k)^{\mu_2} (k_1 - k_2)^{\mu_3} k_1^{\mu_4}}{(k_1 - k)^2 (k_1 - k_2 - k)^2 (k_1 - k_2)^2 k_1^2 k_2^2}$};
    \end{feynman}
\end{tikzpicture}
\end{equation}
where $e$ is the electromagnetic charge, $\gamma_\mu$ is a Dirac matrix with Dirac index $\mu$, and $\mathrm{Tr}_d$ indicates that the traces are $d$-dimensional traces. There is a tremendous amount of structure in the $d$-dimensional momentum integral piece, and a lot of modern research into scattering amplitudes involves simplifying the calculation of this momentum integral \cite{Travaglini:2022uwo,Abreu:2022mfk}. This work will instead focus on the structure of the $d$-dimensional Dirac trace, $\mathrm{Tr}_d$. After performing the momentum integral in \Cref{eq:examplecalc}, what remains is a linear combination over possible contractions $g_{\mu_1 \mu_2} g_{\mu_3 \mu_4}$, $g_{\mu_1 \mu_3} g_{\mu_2 \mu_4}$ and $g_{\mu_1 \mu_4} g_{\mu_2 \mu_3}$. The main observation made in this work is that the defining relations for gamma matrices in $d$-dimensions:
\begin{equation}\label{eq:gdef}
\{\gamma_\mu,\gamma_\nu\} := \gamma_\mu \gamma_\nu + \gamma_\nu \gamma_\mu = 2 g_{\mu \nu}\mathbbm{1}, \qquad g_{\mu \mu} = d
\end{equation}
can be formulated as a `deletion-contraction relation' on an appropriately constructed graph, relating the value of $\mathrm{Tr}_d$ on one graph to its value on a graph where an edge has been deleted, and a graph where the same edge has been contracted. Such a recurrence relation is very common among polynomial graph invariants, and was initially observed for the chromatic polynomial~\cite{Whitney1932} and number of spanning trees~\cite{Brooks1940}. This was later generalised to a two-variable polynomial known as the Tutte polynomial~\cite{Tutte:1954_}, which contains many other graph invariants which satisfy deletion-contraction relations as special cases (as shown in \Cref{fig:1}) such as the reliability polynomial \cite{Chari1998ReliabilityPA}, Jones polynomial of alternating knots represented as graphs \cite{THISTLETHWAITE1987297}, and the partition function of the $q$-state Potts models \cite{Fortuin:1971dw} ($q = 2,3$ shown on \Cref{fig:1}). 

\begin{figure}[t]
\begin{center}
\includegraphics{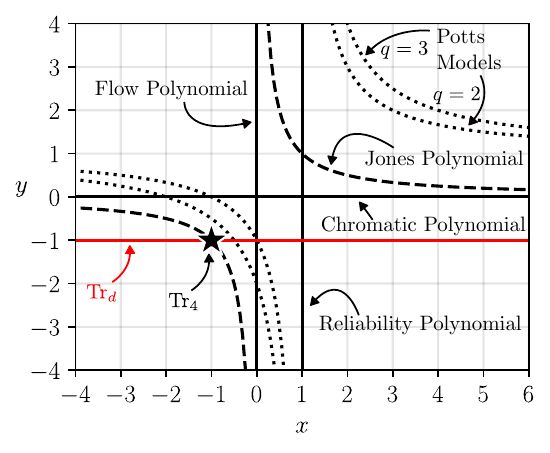}
%\captionsetup{width=.95\linewidth}
\caption{Plot of the various restrictions of the Tutte polynomial $T(x,y)$ in the Tutte plane spanned by the two variables $x$ and $y$. The dimensionally regulated trace $\mathrm{Tr}_d$ arises from evaluations along the line $y = -1$, with the point $(x,y) = (-1,-1)$ corresponding to the four dimensional trace $\mathrm{Tr}_4$. 
}\label{fig:1}
\end{center}
\end{figure}

The main theorem of this work relates the $d$-dimensional trace of a string of Dirac matrices contracted in some arbitrary way to an evaluation of the Tutte polynomial of a corresponding graph along the line $y = -1$, as depicted in \Cref{fig:1}.
\newtheorem*{mainthm}{Theorem \ref{mainthm}}
\begin{mainthm}[reworded]
Suppose the tuple $x = (x_1,x_2,\dots,x_{2n})$ contains the integers $1,\dots,n$ each repeated twice in some order. Then
\begin{equation}\label{eq:mf}
\mathrm{Tr}_d(\gamma_{\mu_{x_1}} \cdots \gamma_{\mu_{x_{2n}}}) =  
4 (-1)^{|E|}(-2)^{n - c(\mathrm{Gr}(x))} \ d^{\hspace{.05cm}c(\mathrm{Gr}(x))} \ T\left(\mathrm{Gr}(x); 1 - \tfrac{d}{2}, -1\right)
\end{equation}
where if $x_1,\dots,x_{2n}$ are placed in order around the circumference of a circle, with straight chords connecting the repeated integers, then $\mathrm{Gr}(x)$ is the graph with a vertex for each chord and an edge connecting pairs of chords that intersect, $c(\mathrm{Gr}(x))$ is the number of connected components of $\mathrm{Gr}(x)$, $|E|$ is the number of edges of $\mathrm{Gr}(x)$, and $T(\mathrm{Gr}(x);x,y)$ is the Tutte Polynomial of $\mathrm{Gr}(x)$ in the variables $x$ and $y$. 
\end{mainthm}
Specialising to $d = 4$, the $4$-dimensional trace operation $\mathrm{Tr}_4$ corresponds to evaluations of the Tutte polynomial at $(x,y) = (-1,-1)$. As it turns out, this is a special point in the Tutte plane, and corresponds to evaluations of the bicycle number \cite{ROSENSTIEHL1978195}:
\newtheorem*{mainthm4}{Theorem \ref{cor}}
\begin{mainthm4}[reworded]
Suppose the tuple $x = (x_1,x_2,\dots,x_{2n})$ contains the integers $1,\dots,n$ each repeated twice in some order. Then
\begin{equation}
    \mathrm{Tr}_4(\gamma_{\mu_{x_1}} \cdots \gamma_{\mu_{x_{2n}}}) = 
    4(-2)^{n+c(\mathrm{Gr}(x))+\mathrm{dim}(B(\mathrm{Gr}(x)))}
\end{equation}
where $c(\mathrm{Gr}(x))$ is the number of connected components of $\mathrm{Gr}(x)$, and $\mathrm{dim}(B(\mathrm{Gr}(x)))$ is the dimension of the bicycle space of $\mathrm{Gr}(x)$. 
\end{mainthm4}

Definitions and details of the proofs of these theorems are presented in \Cref{sec:proof}. As an example application, \Cref{eq:mf} can be used to compute the $d$-dimensional Dirac traces required for the two-loop QED diagram shown in \Cref{eq:examplecalc}:
\begin{align}\label{eq:extr1}
g_{\mu_1 \mu_2} g_{\mu_3 \mu_4} \mathrm{Tr}_d\left(\gamma_{\nu} \gamma_{\mu_1} \gamma_{\mu_5} \gamma_{\mu_2} \gamma_{\nu} \gamma_{\mu_3} \gamma_{\mu_5} \gamma_{\mu_4} \right) = \mathrm{Tr}_d\left(
\wick{\c2 \gamma_{\alpha_1} 
\c1 \gamma_{\alpha_2} 
\c3 \gamma_{\alpha_3}
\c1 \gamma_{\alpha_2}
\c2 \gamma_{\alpha_1}
\c1 \gamma_{\alpha_4}
\c3 \gamma_{\alpha_3}
\c1 \gamma_{\alpha_4}}
\right) \nonumber \\
= \mathrm{Tr}_d \left( \begin{tikzpicture}[scale=0.75,baseline=(a.base)]
\draw (0,0) circle (0.75);
\node (a) at (-0.75,0) {$\bullet$};
\node (b) at (-0.53,0.53) {$\bullet$};
\node (c) at (0,0.75) {$\bullet$};
\node (d) at (0.53,0.53) {$\bullet$};
\node (e) at (0.75,0) {$\bullet$};
\node (f) at (0.53,-0.53) {$\bullet$};
\node (g) at (0,-0.75) {$\bullet$};
\node (h) at (-0.53,-0.53) {$\bullet$};
\draw (a.center) -- (e.center);
\draw (b.center) -- (d.center);
\draw (c.center) -- (g.center);
\draw (f.center) -- (h.center);
\node at (-1.1,0) {$1$};
\node at (-0.8,0.8) {$2$};
\node at (0,1.1) {$3$};
\node at (0.8,0.8) {$2$};
\node at (1.1,0) {$1$};
\node at (0.8,-0.8) {$4$};
\node at (0,-1.1) {$3$};
\node at (-0.8,-0.8) {$4$};
\end{tikzpicture}\right) = 32 d \ T\left(
\begin{tikzpicture}[baseline=(base.base)]
    \begin{feynman}
      \vertex (a) at (0,0);
      \vertex (b) at (0.5,0);
      \vertex (c) at (1,0);
      \vertex (d) at (0.5,-0.5);
      \vertex (base) at (0.5,-0.25);
      \diagram*{
        (a) -- (d) -- (b);
        (c) --(d);
      };
      \node at (a) {$\bullet$};
      \node at (b) {$\bullet$};
      \node at (c) {$\bullet$};
      \node at (d) {$\bullet$};
      \node at (0,0.3) {$1$};
      \node at (0.5,0.3) {$2$};
      \node at (1,0.3) {$4$};
      \node at (0.8,-0.5) {$3$};
    \end{feynman}
\end{tikzpicture}; 1 - \tfrac{d}{2},-1 \right) = -4 d^4+24 d^3-48 d^2+32 d
\end{align}
\begin{align}\label{eq:extr2}
g_{\mu_1 \mu_3} g_{\mu_2 \mu_4}  \mathrm{Tr}_d\left(\gamma_{\nu} \gamma_{\mu_1} \gamma_{\mu_5} \gamma_{\mu_2} \gamma_{\nu} \gamma_{\mu_3} \gamma_{\mu_5} \gamma_{\mu_4} \right) = \mathrm{Tr}_d\left(
\wick{\c1 \gamma_{\alpha_1} 
\c2 \gamma_{\alpha_2} 
\c3 \gamma_{\alpha_3}
\c4 \gamma_{\alpha_4}
\c1 \gamma_{\alpha_1}
\c2 \gamma_{\alpha_2}
\c3 \gamma_{\alpha_3}
\c4 \gamma_{\alpha_4}}
\right) \nonumber \\
= \mathrm{Tr}_d \left( \begin{tikzpicture}[scale=0.75,baseline=(a.base)]
\draw (0,0) circle (0.75);
\node (a) at (-0.75,0) {$\bullet$};
\node (b) at (-0.53,0.53) {$\bullet$};
\node (c) at (0,0.75) {$\bullet$};
\node (d) at (0.53,0.53) {$\bullet$};
\node (e) at (0.75,0) {$\bullet$};
\node (f) at (0.53,-0.53) {$\bullet$};
\node (g) at (0,-0.75) {$\bullet$};
\node (h) at (-0.53,-0.53) {$\bullet$};
\draw (a.center) -- (e.center);
\draw (b.center) -- (f.center);
\draw (c.center) -- (g.center);
\draw (d.center) -- (h.center);
\node at (-1.1,0) {$1$};
\node at (-0.8,0.8) {$2$};
\node at (0,1.1) {$3$};
\node at (0.8,0.8) {$4$};
\node at (1.1,0) {$1$};
\node at (0.8,-0.8) {$2$};
\node at (0,-1.1) {$3$};
\node at (-0.8,-0.8) {$4$};
\end{tikzpicture}\right) = - 32d \ T\left(
\begin{tikzpicture}[baseline=(base.base)]
    \begin{feynman}
      \vertex (a) at (0.25,0.25);
      \vertex (b) at (0.25,-0.25);
      \vertex (c) at (-0.25,0.25);
      \vertex (d) at (-0.25,-0.25);
      \vertex (base) at (0,-0.1);
      \diagram*{
        (a) -- (b) -- (c) -- (d) -- (a);
        (a) -- (c);
        (b) -- (d);
      };
      \node at (a) {$\bullet$};
      \node at (b) {$\bullet$};
      \node at (c) {$\bullet$};
      \node at (d) {$\bullet$};
      \node at (0.45,0.45) {$2$};
      \node at (0.45,-0.45) {$4$};
      \node at (-0.45,0.45) {$1$};
      \node at (-0.45,-0.45) {$3$};
    \end{feynman}
\end{tikzpicture}; 1 - \tfrac{d}{2},-1 \right) = 4 d^4-48 d^3+112 d^2-64 d
\end{align}
The third possible contraction $g_{\mu_1 \mu_4} g_{\mu_2 \mu_3}$ gives the same graph up to isomorphism as in \Cref{eq:extr1}. 
\Cref{eq:extr1,eq:extr2} can be verified by various computer-algebra codes available that perform Dirac Traces~\cite{Jamin:1991dp,Shtabovenko:2023idz,Kublbeck:1990xc,Cyrol:2016zqb,Kuipers:2012rf}. 
In \Cref{sec:trad}, the formulas traditionally used in these computer-algebra codes to simplify traces of Dirac matrices are re-interpreted as relations on Tutte polynomials of graphs. 
Furthermore, the relationship to Tutte polynomials also provides \textit{new} relationships for Dirac traces, which can provide computational benefits when calculating long traces. 
Comparisons between different code packages capable of performing $d$-dimensional Dirac traces in Mathematica~\cite{Mathematica} and computing the Dirac trace with the Tutte polynomial are shown in \Cref{fig:comp}, where traces of strings of Dirac matrices contracted randomly. 
Though these tests were all performed on a single-core machine, 
the observed performance improvement by using specialised Tutte polynomial algorithms suggests that rewriting Dirac traces as Tutte polynomials may be computationally useful. 

\begin{figure}[t]
    \centering
    \hspace{-1.25cm}\includegraphics[width=0.65\linewidth]{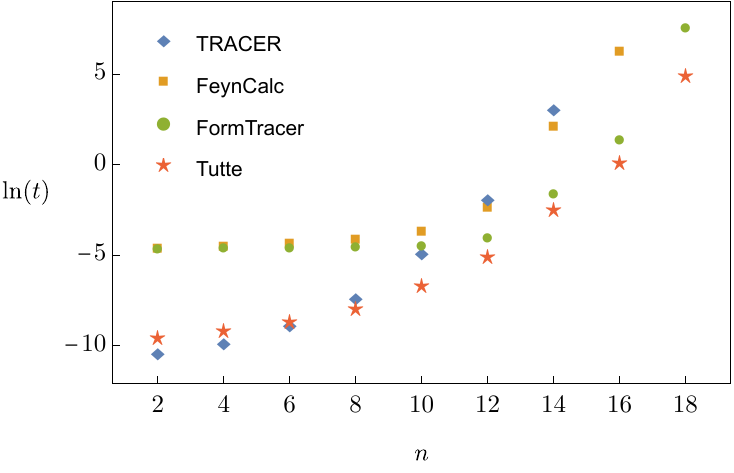}
    \begin{textblock*}{5cm}(6.85cm,3.45cm) % {block width} (coords) 
    {\cite{Jamin:1991dp}}
    \end{textblock*}
    \begin{textblock*}{5cm}(7.2cm,4.05cm) % {block width} (coords) 
    {\cite{Shtabovenko:2023idz,Kublbeck:1990xc}}
    \end{textblock*}
    \begin{textblock*}{5cm}(7.45cm,4.58cm) % {block width} (coords) 
    {\cite{Cyrol:2016zqb,Kuipers:2012rf}}
    \end{textblock*}
    \caption{Average runtime of computing $d$-dimensional Dirac Traces of $2n$ $\gamma$-matrices contracted randomly. Runtime in seconds is plotted on a log-scale using various code packages available in Mathematica \cite{Mathematica} averaged over $128$ different examples for each $n$.
    The red star datapoints labelled `Tutte' refers to a naive implementation of \Cref{eq:mf} with Mathematica's inbuilt TuttePolynomial function. }
    \label{fig:comp}
\end{figure}

The complexity of computing Dirac traces is explored via the connection to Tutte polynomials in \Cref{sec:comp}. 
The problem of computing $\mathrm{Tr}_d$ on strings of $2n$ Dirac matrices contracted in some arbitrary fasion is equivalent to the problem of computing the Tutte polynomial of circle graphs with $n$ vertices along the line $y = -1$ in the Tutte plane (where a circle graph is a graph obtained by drawing straight chords on a common circle, with a vertex for each chord and an edge between chords that intersect). 
In the more general case of evaluating the Tutte polynomial for arbitrary graphs $G$, it is known for example that evaluating $T(G;1-n,0)$ for some $n \in \mathbb{Z}, n \geq 4$ is $\#\P$-complete by parsimonious reduction from $n$-colouring to 3SAT \cite{Jaeger_Vertigan_Welsh_1990}. By utilising some general interpolation theorems proven in Ref~\cite{Jaeger_Vertigan_Welsh_1990} a corollary is that evaluating $T(G;x,y)$ along the line $\{y=-1\}$ is also $\#\P$-hard. 
It might be hoped that a similar strategy would apply to the restricted class of circle graphs, and that Dirac traces are $\#\P$-hard to compute. 
In the case where $G$ is restricted to circle graphs, it is only known that evaluating $T(G;1-n,0)$ for $n \in \mathbb{Z}, n \geq 4$ is $\NP$-hard \cite{Unger} (due to lack of a parsimonious reduction), and we prove the statement:
\newtheorem*{mainthm2}{Theorem \ref{thm:cstar}}
\begin{mainthm2}[reworded]
Let the class $\mathbf{Circle}^*$ be the class of circle graphs and all $k$-stretchings of circle graphs, where a $k$-stretched graph (for $k \in \mathbb{N}, k \geq 2$) is obtained by subdividing each of the edges into $k$-segments. Computing $\{T(G;x,-1) : x \in \mathbb{R}\}$ (expressed as a polynomial in $x$) on $\mathbf{Circle}^*$ is $\NP$-hard. 
\end{mainthm2}
This falls short of the goal of proving that dimensionally regulated Dirac traces are $\#\P$-hard to compute, but the strategy may still be salvageable in future investigations. For the special case of evaluating the $4$-dimensional trace $\mathrm{Tr}_4$ however, this corresponds to evaluating $T(G;-1,-1)$, which happens to be an exceptional point in the Tutte plane, as there is a polynomial-time algorithm that solves this problem:
\newtheorem*{mainthm3}{Theorem \ref{thm:4P}}
\begin{mainthm3}[reworded]
The bicycle number $\mathrm{dim}(B(G))$ can be computed in polynomial time by gaussian elimination, and the $4$-dimensional trace is in $\FP$. 
\end{mainthm3}
where $\FP$ is the class of functions computable in polynomial time. 
Note that fixed-order perturbative QFT calculations will only require the Dirac trace to some finite power of $(d - 4)$, in other words, the taylor-series coefficients of the Tutte polynomial about the point $(-1,-1)$ along the line $y = -1$. 
These taylor-series coefficients might have graph theoretic interpretations that allow for polynomial-time evaluations, but this is not explored in this work. 
Details, definitions and proofs of \Cref{thm:cstar,thm:4P} are located in \Cref{sec:comp}.
Finally, it is well understood that Dirac traces with open (uncontracted) indices can be reduced to Dirac traces with all indices contracted. 
This allows for products of Dirac traces to be computed, as well as traces including 't Hooft Veltman $\gamma_5$. As these extensions mostly rely on connection to Tutte polynomials for single traces and other formulas known in the literature, a brief perspective on this topic is presented in \Cref{sec:4}.

\section{Mapping between Dirac Trace and Tutte Polynomial}\label{sec:proof}

In this section, the formal connection between the dimensionally regulated Dirac trace and the Tutte polynomial is proven. Before the main theorems and proofs are presented, some groundwork is necessary to introduce key ideas and fix notation. We start by formalising the dimensionally regulated Dirac Trace. One approach to constructing the $d$-dimensional trace $\mathrm{Tr}_d$ is by exhibiting explicit infinite-dimensional representations of the $d$-dimensional Clifford Algebra, as discussed in Ref~\cite{Collins:1984xc}. In this section, a formal construction of $\mathrm{Tr}_d$ will be presented instead by considering the defining relations 
\begin{equation}\label{eq:gdefs}
\{\gamma_\mu, \gamma_\nu\} := \gamma_\mu \gamma_\nu + \gamma_\nu \gamma_\mu = g_{\mu \nu}\mathbbm{1}, \qquad g_{\mu \mu} = d
\end{equation}
as relations on symbols, with no reference made to matrix-like structures for the $\gamma$-objects. For the purposes of this work, no distinction will be made between lower and upper indices (all indices will be lower indices) as everything will be considered in Euclidean signature. 

For concreteness, first consider the case of computing a single Dirac trace where there are no free indices. Extensions to traces with open (uncontracted) Dirac indices, and Dirac traces including $\gamma_5$ insertions are presented in \Cref{sec:4}. Let $\Sigma$ be a countable alphabet of symbols (usually labelled as $\Sigma = \{\mu_1,\mu_2,\dots\}$), and let $P(n)$ be the set of all permutations of tuples of length $2n$ with symbols from $\Sigma$, such that each symbol appears exactly twice. 
The map $\mathrm{Tr}_d$ will be defined as a map from $x = (x_1,\dots,x_{2n}) \in P(n)$ to $\mathrm{Tr}_d(x) \in \mathbb{Z}[d]$, the ring of polynomials in $d$ with integer coefficients, with the interpretation that $\mathrm{Tr}_d(x) = \mathrm{Tr}_d(\gamma_{x_1} \cdots \gamma_{x_{2n}})$. 

\begin{definition}\label{defn:Trd}
A function $\mathrm{Tr}_d : \cup_{n \in \mathbb{N}_{\geq 0}} P(n) \to \mathbb{Z}[d]$ is called a $d$-dimensional trace operation if it satisfies the base case $\mathrm{Tr}(\emptyset) = 4$ where $\emptyset \in P(0)$ is the empty tuple, and for all $n \geq 1$, $1 \leq i < 2n$ and $x = (x_1,\cdots,x_{2n}) \in P(n)$, it satisfies the recursion relations:% \Cref{eq:rec1,eq:rec2}:
\begin{align}\label{eq:rec1}
x_i \neq x_{i+1} &\quad \implies \quad \mathrm{Tr}_d(x) = - \mathrm{Tr}_d(S_{i,i+1}(x)) + 2 \mathrm{Tr}_d(C_{i,i+1}(x)),\\ \label{eq:rec2}
x_i = x_{i+1} &\quad \implies \quad \mathrm{Tr}_d(x) = d \ \mathrm{Tr}_d(C_{i,i+1}(x)),
\end{align}
where $S_{i,j} : P(n) \to P(n)$ is the map that swaps the $i$-th element $x_i$ with the $j$-th element $x_j$ in the tuple, and $C_{i,j} : \{x \in P(n)\} \to P(n-1)$ is the map that removes the $i$-th and $j$-th elements in the tuple, and relabels all other occurrences of $x_j$ in the tuple to $x_i$. 
\end{definition}
\Cref{defn:Trd} is a formalisation of the defining relations given in \Cref{eq:gdefs}. The base case $\mathrm{Tr}_d(\emptyset) = 4$ is the conventional choice used in most calculations involving dimensional regularisation, but this condition can be changed, for instance to $\mathrm{Tr}_d(\emptyset) = 2^{\frac{d}{2}}$. This latter choice may seem more natural given that in even spacetime dimensions $d$, the standard representation of the Clifford Algebra has dimension $2^{\frac{d}{2}}$. However, such modifications to the base case do not affect the renormalization procedure in any meaningful way, and simply amount to a shift in the renormalization constants \cite{Collins:1984xc}. 

A graph $G = (V,E,\pi)$ refers to an undirected graph possibly with multi-edges and loops, and is specified by a vertex set $V$, an edge set $E$, and a map $\pi : E \to (V \times V) / \sim_f$ associating an edge with its endpoints, where for all $v_0,v_1 \in V$, $(v_0,v_1) \sim_f (v_1,v_0)$. The collection of all graphs is typeset as \textbf{Graph}. A circle graph is a graph formed by first drawing a collection of straight chords on a common circle, then associating each chord with a vertex of the graph, and an edge between two vertices if and only if the two corresponding chords intersect. The collection of all circle graphs is typeset as \textbf{Circle}.

\begin{definition}
The map 
\begin{equation}
    \mathrm{Gr} : \bigcup_{n \in \mathbb{N}_{\geq 0}} P(n) \to \mathbf{Circle}
\end{equation}
associates to every $x \in P(n)$ a circle graph. The construction first places the elements of the tuple $x$ in order around a circle, joining the repeated elements with straight chords. The circle graph is the intersection graph of this construction, so that the vertex set is the set of symbols appearing in the tuple $x$, and edges correspond to intersections of chords. 
\end{definition}

As an example:
\begin{center}
\begin{tikzpicture}
\node at (-4.5,0) {$P(4) \ni (\mu_1,\mu_2,\mu_3,\mu_2,\mu_1,\mu_4,\mu_3,\mu_4) \mapsto $};
\draw (0,0) circle (0.75);
\node (a) at (-0.75,0) {$\bullet$};
\node (b) at (-0.53,0.53) {$\bullet$};
\node (c) at (0,0.75) {$\bullet$};
\node (d) at (0.53,0.53) {$\bullet$};
\node (e) at (0.75,0) {$\bullet$};
\node (f) at (0.53,-0.53) {$\bullet$};
\node (g) at (0,-0.75) {$\bullet$};
\node (h) at (-0.53,-0.53) {$\bullet$};
\draw (a.center) -- (e.center);
\draw (b.center) -- (d.center);
\draw (c.center) -- (g.center);
\draw (f.center) -- (h.center);
\node at (-1.04,0) {$\mu_1$};
\node at (-0.74,0.74) {$\mu_2$};
\node at (0,1) {$\mu_3$};
\node at (0.74,0.74) {$\mu_2$};
\node at (1.04,0) {$\mu_1$};
\node at (0.74,-0.74) {$\mu_4$};
\node at (0,-1) {$\mu_3$};
\node at (-0.74,-0.74) {$\mu_4$};
\node at (1.5,0) {$\mapsto$};
\node (a) at (0+2  ,0+0.25) {$\bullet$};
\node (b) at (0.5+2,0+0.25) {$\bullet$};
\node (c) at (1+2  ,0+0.25) {$\bullet$};
\node (d) at (0.5+2,-0.5+0.25) {$\bullet$};
\draw (a.center) -- (d.center) -- (b.center);
\draw (c.center) --(d.center);
\node at (0+2,0.2+0.35) {$\mu_1$};
\node at (0.5+2,0.2+0.35) {$\mu_2$};
\node at (1+2,0.2+0.35) {$\mu_4$};
\node at (0.8+2,-0.5+0.25) {$\mu_3$};
\node at (4,0) {$\in \mathbf{Circle}$};
\end{tikzpicture}
\end{center}

Clearly this construction satisfies a cyclic symmetry, where the same graph is associated to all cyclic shifts of a tuple in $P(n)$. This observation will be related to the property that traces are invariant under cyclic shifts. Furthermore, the associated graph is invariant under reversing the order of the tuple, which is related to invariance of traces under transposes and inverses. 

In the general case, the Tutte polynomial is a polynomial defined for matroids (this will be relevant in \Cref{sec:comp}, where a brief introduction to matroids is presented). For the purposes of this section, it is simpler to restrict to the special case of undirected graphs. To define the Tutte polynomial, we first introduce some basic graph terminology: 
\begin{itemize}
\item[$\bullet$] $c(G)$ is the number of connected components of the graph. 
\item[$\bullet$] $b(G)$ is the number of bridges of the graph, where a bridge is an edge that if removed from $G$, the number of connected components increases by one. 
\item[$\bullet$] $l(G)$ is the number of loops of the graph $G$, where a loop is an edge that has both of it's endpoints as the same vertex. 
\item[$\bullet$] $G - e$ is the graph obtained by deleting the edge $e$ from the graph $G$
\item[$\bullet$] $G/e$ is the graph obtained by contracting the edge $e$ (deleting $e$, and merging it's two endpoints).
\end{itemize}

\begin{definition}{\textup{\cite{Tutte:1954_}}}
The Tutte polynomial is the unique function $T : \mathbf{Graph} \to \mathbb{Z}[x,y]$ satisfying the base case $T(G) = x^{|b(G)|}y^{|l(G)|}$ if $G$ has no edges that are not bridges or loops, and for any edge $e$ that is not a bridge or a loop, the deletion-contraction relation holds:
\begin{equation}
T(G) =  T(G - e) + T(G/ e).
\end{equation}
\end{definition}
The main observation of this work is that the anticommutation relations of the Dirac algebra can be reinterpreted as a deletion-contraction relation for the Tutte polynomial computed on the associated circle graph. To see this, it is simplest to make some technical detours. Firstly, the `recipe theorem' for the Tutte polynomial for graphs which was originally stated and proved in Ref~\cite{Oxley:1979_} is given below in the form presented in Ref~\cite{Bollobs1998}:
\begin{prop}{\textup{\cite[Thm. 1]{Oxley:1979_}}}\label{propox}
Suppose the function $\overline{T} : \mathbf{Graph} \to \mathbb{Z}[x,y,\alpha,\sigma,\tau]$ satisfies the base case $\overline{T}(G) = x^{|b(G)|} y^{|l(G)|} \alpha^{|V|}$ if $G$ has no edges that are not bridges or loops, and satisfies a modified recurrence relation
\begin{equation}
\overline{T}(G) = \sigma \overline{T}(G - e) + \tau \overline{T}(G/e).
\end{equation}
Then $\overline{T}$ is given by a transformation of variables of the Tutte polynomial:
\begin{equation}
\overline{T}(G;x,y,\alpha,\sigma,\tau) = \alpha^{c(G)} \sigma^{|E|-|V|+c(G)} \tau^{|V|-c(G)} T\left(G;\frac{\alpha x}{\tau},\frac{y}{\sigma}\right).
\end{equation}
\end{prop}

In the following proofs, it is easiest to utilise an alternative parametrisation of the Tutte polynomial:

\begin{definition}\label{defn:gc}
For a graph $G$, let $f : V \to \{1,\cdots,n\}$ be an arbitrary function. Assigning an arbitrary direction to each edge of the graph, where $\mathrm{src},\mathrm{dest} : E \to V$  are functions associating an edge with its source vertex and its destination vertex, then the number of collisions is defined by
\begin{equation}
\mathrm{coll}(f) = \sum_{e \in E} \delta_{f(\mathrm{src}(e)),f(\mathrm{dest}(e))}
\end{equation}
where $\delta_{ab}$ is the kronecker-delta. The generalised chromatic-polynomial is a function $\chi : \mathbf{Graph} \to \mathbb{Z}[q,n]$ given by
\begin{equation}
\chi(G;q,n) = \sum_{f : V \to \{1,\cdots,n\}} q^{\mathrm{coll}(f)}
\end{equation}
\end{definition}

Note that $\chi(G;0,n)$ is the chromatic polynomial of the graph $G$, that counts the number of proper colourings of the graph $G$ with $n$-colours (where the endpoints of every edge must be coloured different colours). The generalised chromatic-polynomial satisfies a deletion-contraction relation, which allows us to prove that it can be written as an evaluation of the Tutte polynomial:

\begin{prop}\label{prop:b}
\begin{equation}\label{eq:16}
\chi (G;q,n) = \overline{T}\left(G;1 + \frac{q-1}{n},q,n,1,q-1\right) = n^{c(G)} (q-1)^{|V|-c(G)} T\left(G; 1 + \frac{n}{q-1}, q\right)
\end{equation}
\end{prop}

\begin{proof}
Note that if $G$ has no edges, then $\chi(G;q,n) = n^{|V|}$ as there can be no collisions. If $e$ is a bridge, then $\chi(G;q,n)$ is composed of a contribution corresponding to the endpoints of $e$ being colored the same $\left(\frac{q}{n} \right) \chi(G - e; q,n)$, and a contribution corresponding to the endpoints of $e$ being colored differently $\left(\frac{n-1}{n}\right) \chi(G - e; q,n)$, giving a relationship $\chi(G;q,n) = \left(1 + \frac{q-1}{n}\right) \chi(G - e; q,n)$. If $e$ is a loop, then $\chi(G;q,n) = q \chi(G-e;q,n)$ as the loop always gives a collision. Finally, for any other edge $e$ there is a deletion-contraction relation, described pictorially in an example:
\begin{equation}\label{eq:17}
\begin{tikzpicture}[baseline=-20pt]
\node (za) at (-0.7,-0.6) {$\chi\bigg($};
\node (zb) at (2.4,-0.6) {$;q,n\bigg) = \chi \bigg($};
\node (zc) at (6.5,-0.6) {$;q,n\bigg) + (q - 1) \ \chi \bigg($};
\node (zd) at (10.25,-0.6) {$;q,n \bigg)$};
\node (a) at (0,0) {$\bullet$};
\node (b) at (1,0) {$\bullet$};
\node (c) at (0.5,-1.2) {}; %ellipse center
\node (aa) at (-0.2,-1.03) {};
\node (ab) at (0.1,-0.9) {};
\node (ac) at (0.4,-0.83) {};
\node (ba) at (0.9,-0.9) {};
\node (bb) at (1.2,-1.03) {};
\node at (0.5,-0.2) {$e$};

\draw (a.center) to (b.center);
\draw (a.center) to [bend left = 50] (b.center);
\draw (a.center) to (aa.center);
\draw (a.center) to (ab.center);
\draw (a.center) to (ac.center);
\draw (b.center) to (ba.center);
\draw (b.center) to (bb.center);

\coordinate (s) at (3.7,0);

\node (sa) at ($(a.center)+(s.center)$) {$\bullet$};
\node (sb) at ($(b.center)+(s.center)$) {$\bullet$};
\draw ($(a.center)+(s.center)$) to [bend left = 50] ($(b.center)+(s.center)$);
\draw ($(a.center)+(s.center)$) to ($(aa.center)+(s.center)$);
\draw ($(a.center)+(s.center)$) to ($(ab.center)+(s.center)$);
\draw ($(a.center)+(s.center)$) to ($(ac.center)+(s.center)$);
\draw ($(b.center)+(s.center)$) to ($(ba.center)+(s.center)$);
\draw ($(b.center)+(s.center)$) to ($(bb.center)+(s.center)$);

\coordinate (t) at (8.45,0);
\node (ta) at ($(a.center)+(t.center)+(0.5,0)$) {$\bullet$};
\draw (ta.center) to [out=45,in=135,looseness=200] ($(ta.center)+(0.01,0)$);
\draw (ta.center) to ($(aa.center)+(t.center)$);
\draw (ta.center) to ($(ab.center)+(t.center)$);
\draw (ta.center) to ($(ac.center)+(t.center)$);
\draw (ta.center) to ($(ba.center)+(t.center)$);
\draw (ta.center) to ($(bb.center)+(t.center)$);

\draw [fill=gray!20] (c) ellipse (0.8cm and 0.4cm);
\draw [fill=gray!20] ($(c.center)+(s.center)$) ellipse (0.8cm and 0.4cm);
\draw [fill=gray!20] ($(c.center)+(t.center)$) ellipse (0.8cm and 0.4cm);
\end{tikzpicture}
\end{equation}
where the grey blob represents the rest of the unshown graph that has arbitrary structure. The first term on the right hand side of \Cref{eq:17} correctly counts the contributions from the sum where the two vertices connected by $e$ are coloured differently, and the second term corrects for the behaviour when the two vertices are coloured the same (in order to generate the factor of $q$ due to the collision). \Cref{eq:16} follows from the above observations and an application of \Cref{propox}. 
\end{proof}

\begin{prop}\label{prop:a}
There is a unique function $\mathrm{Tr}_d$ satisfying the postulates outlined in \Cref{defn:Trd}, and it is given by:
\begin{equation}\label{eq:step1}
\mathrm{Tr}_d(x) = 4 (-1)^{|E|}  \ \chi(\mathrm{Gr}(x);-1,d)
\end{equation}
\end{prop}

\begin{proof}
For positive integer values of $d$, the right hand side of \Cref{eq:step1} (RHS) can be expanded using \Cref{defn:gc} as
\begin{align}
\mathrm{RHS}(\mathrm{Gr}(x)) &= 4 (-1)^{|E|} \sum_{f : V(\mathrm{Gr}(x)) \to \{1,\dots,d\}} (-1)^{\mathrm{coll}(f)} \\
&= 4 \sum_{f : V(\mathrm{Gr}(x)) \to \{1,\dots,d\}} \prod_{e \in E} (-1)^{1 - \delta_{f(\mathrm{src}(e)),f(\mathrm{dest}(e))}}
\end{align}
We will now proceed to show that $\mathrm{RHS}$ satisfies the defining postulates of $\mathrm{Tr}_d$. Since $\mathrm{Gr}(\emptyset)$ is the empty graph, $\mathrm{RHS}(\mathrm{Gr}(x)) = 4$, thus the base condition is satisfied. Now suppose that $x = (x_1,\dots,x_{2n}), 1 \leq i < 2n$, and $x_i = x_{i+1}$. The vertex corresponding to $(x_i, x_{i+1})$ in $\mathrm{Gr}(x)$ is an isolated vertex, and $\mathrm{RHS}(\mathrm{Gr}(x)) = d\cdot \mathrm{RHS}(\mathrm{Gr}(C_{i,i+1}(x)))$ as the isolated vertex can be coloured in $d$ ways and there are no associated collisions. Hence the condition corresponding to $\gamma_\mu \gamma_\mu = d \cdot \mathbbm{1}$ is satisfied. 

\newcommand*\circled[1]{\tikz[baseline=(char.base)]{
\node[shape=circle,draw,color=red,fill=white,inner sep=1pt] (char) {#1};}}

Suppose instead that $x_i \neq x_{i+1}$. Then $x_i$ and $x_{i+1}$ correspond to two different vertices in $\mathrm{Gr}(x)$, which are either connected by a single edge, or not connected by any edges. 
Suppose $j,k$ are chosen such that $i,i+1,j,k$ are four different integers. 
Then, the following equation holds regardless of the structure of the chords connecting the regions $\circled{1},\circled{2}$ and $\circled{3}$ (not shown below):

\begin{equation}\label{eq:circs}
\begin{tikzpicture}[baseline=0pt]
\draw [black,thick,domain=60:120] plot ({cos(\x)}, {sin(\x)});
\draw [black,thick,domain=180:240] plot ({cos(\x)}, {sin(\x)});
\draw [black,thick,domain=-60:0] plot ({cos(\x)}, {sin(\x)});
\draw [red,thick,domain=120:180] plot ({cos(\x)}, {sin(\x)});
\draw [red,thick,domain=240:300] plot ({cos(\x)}, {sin(\x)});
\draw [red,thick,domain=0:60] plot ({cos(\x)}, {sin(\x)});
\node (aa) at ({cos(80)}, {sin(80)}) {$\bullet$};
\node (ab) at ({cos(100)}, {sin(100)}) {$\bullet$};
\node (ac) at ({cos(210)}, {sin(210)}) {$\bullet$};
\node (ad) at ({cos(-30)}, {sin(-30)}) {$\bullet$};
\draw (aa.center) to (ac.center);
\draw (ab.center) to (ad.center);
\node (ae) at ( 0.5, 1.25) {$x_{i+1}$};
\node (af) at (-0.5, 1.25) {$x_i$};
\node (ag) at (1.15,-0.75) {$x_j$};
\node (ah) at (-1.15,-0.75) {$x_k$};
\node (ai) at (-1.9,0) {$\mathrm{RHS} \bigg($};
\node (aj) at (2.25,0) {$\bigg) + \mathrm{RHS} \bigg($};
\node  at ({cos(150)},{sin(150)}) {\circled{1}};
\node  at ({cos(270)},{sin(270)}) {\circled{3}};
\node  at ({cos(30)},{sin(30)}) {\circled{2}};
\draw [red,thick] ({0.9*cos(0)},{0.9*sin(0)}) to ({1.1*cos(0)},{1.1*sin(0)});
\draw [red,thick] ({0.9*cos(60)},{0.9*sin(60)}) to ({1.1*cos(60)},{1.1*sin(60)});
\draw [red,thick] ({0.9*cos(120)},{0.9*sin(120)}) to ({1.1*cos(120)},{1.1*sin(120)});
\draw [red,thick] ({0.9*cos(180)},{0.9*sin(180)}) to ({1.1*cos(180)},{1.1*sin(180)});
\draw [red,thick] ({0.9*cos(240)},{0.9*sin(240)}) to ({1.1*cos(240)},{1.1*sin(240)});
\draw [red,thick] ({0.9*cos(300)},{0.9*sin(300)}) to ({1.1*cos(300)},{1.1*sin(300)});
\node at (-0.55,0.2) {$a$};
\node at (0.55,0.2) {$b$};

\coordinate (s) at (4.5,0);
\draw [black,thick,domain=60:120] plot  ({4.5+cos(\x)}, {sin(\x)});
\draw [black,thick,domain=180:240] plot ({4.5+cos(\x)}, {sin(\x)});
\draw [black,thick,domain=-60:0] plot   ({4.5+cos(\x)}, {sin(\x)});
\draw [red,thick,domain=120:180] plot   ({4.5+cos(\x)}, {sin(\x)});
\draw [red,thick,domain=240:300] plot   ({4.5+cos(\x)}, {sin(\x)});
\draw [red,thick,domain=0:60] plot      ({4.5+cos(\x)}, {sin(\x)});
\node (ba) at ($({cos(80)}, {sin(80)})+(s)$) {$\bullet$};
\node (bb) at ($({cos(100)}, {sin(100)})+(s)$) {$\bullet$};
\node (bc) at ($({cos(210)}, {sin(210)})+(s)$) {$\bullet$};
\node (bd) at ($({cos(-30)}, {sin(-30)})+(s)$) {$\bullet$};
\draw (ba.center) to (bd.center);
\draw (bb.center) to (bc.center);
\node (be) at ($( 0.5, 1.25)+(s)$) {$x_{i+1}$};
\node (bf) at ($(-0.5, 1.25)+(s)$) {$x_i$};
\node (bg) at ($(1.15,-0.75)+(s)$) {$x_j$};
\node (bh) at ($(-1.15,-0.75)+(s)$) {$x_k$};
\node (bi) at ($(2.25,0) + (s.center)$) {$\bigg) = 2 \mathrm{RHS} \bigg($};
\node  at ({4.5+cos(150)},{sin(150)}) {\circled{1}};
\node  at ({4.5+cos(270)},{sin(270)}) {\circled{3}};
\node  at ({4.5+cos(30)},{sin(30)}) {\circled{2}};
\draw [red,thick] ({4.5+0.9*cos(0)},{0.9*sin(0)}) to ({4.5+1.1*cos(0)},{1.1*sin(0)});
\draw [red,thick] ({4.5+0.9*cos(60)},{0.9*sin(60)}) to ({4.5+1.1*cos(60)},{1.1*sin(60)});
\draw [red,thick] ({4.5+0.9*cos(120)},{0.9*sin(120)}) to ({4.5+1.1*cos(120)},{1.1*sin(120)});
\draw [red,thick] ({4.5+0.9*cos(180)},{0.9*sin(180)}) to ({4.5+1.1*cos(180)},{1.1*sin(180)});
\draw [red,thick] ({4.5+0.9*cos(240)},{0.9*sin(240)}) to ({4.5+1.1*cos(240)},{1.1*sin(240)});
\draw [red,thick] ({4.5+0.9*cos(300)},{0.9*sin(300)}) to ({4.5+1.1*cos(300)},{1.1*sin(300)});
\node at (4.1,0.2) {$a$};
\node at (4.9,0.2) {$b$};

\coordinate (t) at (9,0);
\draw [black,thick,domain=60:120] plot  ({9+cos(\x)}, {sin(\x)});
\draw [black,thick,domain=180:240] plot ({9+cos(\x)}, {sin(\x)});
\draw [black,thick,domain=-60:0] plot   ({9+cos(\x)}, {sin(\x)});
\draw [red,thick,domain=120:180] plot   ({9+cos(\x)}, {sin(\x)});
\draw [red,thick,domain=240:300] plot   ({9+cos(\x)}, {sin(\x)});
\draw [red,thick,domain=0:60] plot      ({9+cos(\x)}, {sin(\x)});
\node (cc) at ($({cos(210)}, {sin(210)})+(t)$) {$\bullet$};
\node (cd) at ($({cos(-30)}, {sin(-30)})+(t)$) {$\bullet$};
\draw (cc.center) to (cd.center);
\node (cg) at ($(1.15,-0.75)+(t)$) {$x_j$};
\node (ch) at ($(-1.15,-0.75)+(t)$) {$x_k$};
\node (ci) at ($(1.5,0) + (t.center)$) {$\bigg)$};
\node  at ({9+cos(150)},{sin(150)}) {\circled{1}};
\node  at ({9+cos(270)},{sin(270)}) {\circled{3}};
\node  at ({9+cos(30)},{sin(30)}) {\circled{2}};
\draw [red,thick] ({9+0.9*cos(0)},{0.9*sin(0)}) to ({9+1.1*cos(0)},{1.1*sin(0)});
\draw [red,thick] ({9+0.9*cos(60)},{0.9*sin(60)}) to ({9+1.1*cos(60)},{1.1*sin(60)});
\draw [red,thick] ({9+0.9*cos(120)},{0.9*sin(120)}) to ({9+1.1*cos(120)},{1.1*sin(120)});
\draw [red,thick] ({9+0.9*cos(180)},{0.9*sin(180)}) to ({9+1.1*cos(180)},{1.1*sin(180)});
\draw [red,thick] ({9+0.9*cos(240)},{0.9*sin(240)}) to ({9+1.1*cos(240)},{1.1*sin(240)});
\draw [red,thick] ({9+0.9*cos(300)},{0.9*sin(300)}) to ({9+1.1*cos(300)},{1.1*sin(300)});
\node at (9,-0.3) {$c$};

\end{tikzpicture}
\end{equation}
where it is understood that $\mathrm{RHS}$ is acting on the corresponding intersection graphs of the circle graphs depicted in \Cref{eq:circs}. 
To prove \Cref{eq:circs}, first notice that for a particular coloring of the chords, if chord $a$ is colored differently to chord $b$, then the contribution from that coloring vanishes on the left hand side of \Cref{eq:circs} as the first circle diagram has a relative factor of $(-1)$ compared to the second circle diagram due to the color collision. So, the only colorings that contribute on the left hand side are colorings in which both chords $a$ and $b$ are colored the same, which are in bijective correspondence with colorings of the chords on the right hand side diagram, where chord $c$ is coloured the same as the chords $a$ and $b$. Note that any chords that connect region $\circled{1}$ to $\circled{3}$, or region $\circled{2}$ to $\circled{3}$ will cross exactly one of either chord $a$ or chord $b$ for the left hand diagrams, and will cross chord $c$ on the right hand diagram. Thus, any such chord will contribute the same factor to the sum on either side of \Cref{eq:circs}. And if a chord connects region $\circled{1}$ to $\circled{2}$, it crosses both chords $a$ and $b$ on the left hand side, but since $a$ and $b$ are colored the same, it is equivalent to not crossing either (as it picks up two minus signs if it is colored differently to $a$ and $b$), as in the right diagram. Hence \Cref{eq:circs} is true, proving that RHS satisfies the recurrence relation corresponding to $\{\gamma_\mu, \gamma_\nu\} = 2 g_{\mu \nu}$, for positive integer values of $d$. 

The observations above show that RHS satisfies the recurrence relations defining $\mathrm{Tr}_d$ for all positive integer values of $d$. Since RHS are finite degree polynomials in $d$, by lagrange interpolation in fact RHS satisfies the recurrence relations for all values of $d$. It is the unique valid definition, as the value of $\mathrm{Tr}_d$ on any particular value of $x$ is uniquely determined by the recurrence relations, which can always simplify $\mathrm{Gr}(x)$ to graphs with no edges.  

\end{proof}

The proof above is formalising the intuition that $d$-dimensional traces of Dirac matrices behave as if there were $d$-different anticommuting matrices for positive integer values of $d$. With this proposition, the main theorem follows:

\begin{theorem}\label{mainthm}
The unique map $\mathrm{Tr}_d$ satisfying the defining properties in \Cref{defn:Trd} is given by:
\begin{equation}\label{eq:mainthm}
\mathrm{Tr}_d(x) =  4 (-1)^{|E|} (-2)^{|V| - c(\mathrm{Gr}(x))} \ d^c \ T\left(\mathrm{Gr}(x), 1 - \tfrac{d}{2}, -1\right)
\end{equation}
where $c(\mathrm{Gr}(x))$ is the number of connected components of $\mathrm{Gr}(x)$, $V$ is the set of vertices of $\mathrm{Gr}(x)$, $E$ is the set of edges of $\mathrm{Gr}(x)$. 
\end{theorem}

\begin{proof}
\Cref{eq:mainthm} follows from \Cref{prop:a} and \Cref{prop:b}. 
\end{proof}

An immediate observation from \Cref{mainthm} is that traces in $d = 4$ correspond to evaluating the Tutte polynomial $T(G;-1,-1)$ for appropriately constructed graphs $G$. This point $(x,y) = (-1,-1)$ is one of the special points in the Tutte plane where the Tutte polynomial is known to be calculable in polynomial time. Following the definitions in \cite{ROSENSTIEHL1978195}, consider $2^V,2^E$ as the free vector spaces over $\mathbb{Z}/2\mathbb{Z}$ with basis given by the set of vertices $V$, and the set of edges $E$ respectively for a graph $G = (V,E)$. There are boundary $\partial$ and coboundary $\delta$ maps defined between the two vector spaces:
\begin{equation}
\begin{tikzcd}[column sep=large]
     2^V \rar["\delta", shift left] &  2^E \lar["\partial", shift left]
\end{tikzcd}
\end{equation}
where $\partial(e)$ is the sum of the two endpoints of edge $e$, $\delta(v)$ is the sum over all edges incident to $v$ (loops are counted twice, thus they don't appear in the result as the base field is $\mathbb{Z}/2\mathbb{Z}$), and the maps are extended by linearity to be well defined on all of $2^V$ and $2^E$. The subspace $\mathrm{ker}(\partial) \cap \mathrm{im}(\delta) \subseteq 2^E$ is the space of bicycles. 

\begin{theorem}[\cite{ROSENSTIEHL1978195}, Thm. 9.1]\label{thm:bic} If a graph $G$ has $|E|$ edges and the dimension of the bicycle space of $G$ is $\mathrm{dim}(B(G))$, then
\begin{equation}
T(G;-1,-1) = (-1)^{|E|} (-2)^{\mathrm{dim}(B(G))}
\end{equation}
\end{theorem}

\begin{theorem}\label{cor}
For $x \in P(n)$, the $4$-dimensional trace is given by 
\begin{equation}
\mathrm{Tr}_4(x) = 4 (-2)^{|V|+c(\mathrm{Gr}(x))+\mathrm{dim}(B(\mathrm{Gr}(x)))}
\end{equation}
where $c(\mathrm{Gr}(x))$ is the number of connected components of $\mathrm{Gr}(x)$, $V$ is the set of vertices of $\mathrm{Gr}(x)$, and $\mathrm{dim}(B(G))$ is the dimension of the bicycle space of $G$. 
\end{theorem}

\begin{proof}
Follows from \Cref{thm:bic,mainthm}. 
\end{proof}

\section{Comparisons to traditional algorithms}\label{sec:trad}

In light of the connection between Dirac traces and Tutte polynomial evaluations, standard identities used to simplify Dirac traces can be reinterpreted as identities for Tutte polynomials. This exercise can be instructive in understanding the structure and nature of currently used algorithms. For instance, both TRACER \cite{Jamin:1991dp} and FORM \cite{Kuipers:2012rf} make use of variations of the following identity to simplify $d$-dimensional traces:

\begin{align}\label{eq:drel}
\gamma_\mu \gamma_{\nu_1} \cdots \gamma_{\nu_n} \gamma_\mu = (-1)^n d\  \gamma_{\nu_1} \cdots \gamma_{\nu_n} + 2 \sum_{j=1}^n (-1)^{n-j} \gamma_{\nu_j} \gamma_{\nu_1} \cdots \gamma_{\nu_{j-1}} \gamma_{\nu_{j+1}} \cdots \gamma_{\nu_n}
\end{align}

{\noindent}where the $\gamma_\mu$ at the right end of the chain on the left hand side of the equation is anticommuted through the chain all the way to the left side. This can be interpreted as an identity for Tutte polynomials:

\begin{align}
&T \left( 
\begin{tikzpicture}[baseline=-20pt]
\node (a) at (0.5,0) {$\bullet$};
\node (c) at (0.5,-1.2) {}; %ellipse center
\draw [fill=gray!20] (c) ellipse (0.8cm and 0.4cm);
\node (aa) at (-0.2,-1.03) {$\bullet$};
\node (ab) at (0.1,-0.87) {$\bullet$};
\node (ac) at (0.4,-0.82) {$\bullet$};
\node (ad) at (1.2,-1.03) {$\bullet$};
\node (ae) at (0.69,-0.6) {$\cdots$};
\node (za) at (-0.1,-1.2) {${\scriptstyle{\nu_1}}$};
\node (zb) at (0.2,-1.05) {${\scriptstyle{\nu_2}}$};
\node (zc) at (0.55,-0.99) {${\scriptstyle \nu_3}$};
\node (zd) at (1.1,-1.15) {${\scriptstyle \nu_n}$};
\node (ba) at (0.5,0.2) {$\mu$};
\draw (a.center) to (aa.center);
\draw (a.center) to (ab.center);
\draw (a.center) to (ac.center);
\draw (a.center) to (ad.center);
\end{tikzpicture}; x,y
\right) = x \cdot T\left(
\begin{tikzpicture}[baseline=-35pt]
\node (c) at (0.5,-1.2) {}; %ellipse center
\draw [fill=gray!20] (c) ellipse (0.8cm and 0.4cm);
\node (aa) at (-0.2,-1.03) {$\bullet$};
\node (ab) at (0.1,-0.87) {$\bullet$};
\node (ac) at (0.4,-0.82) {$\bullet$};
\node (ad) at (1.2,-1.03) {$\bullet$};
\node [rotate=-10] (ae) at (0.85,-0.77) {$\cdots$};
\node (za) at (-0.1,-1.2) {${\scriptstyle{\nu_1}}$};
\node (zb) at (0.2,-1.05) {${\scriptstyle{\nu_2}}$};
\node (zc) at (0.55,-0.99) {${\scriptstyle \nu_3}$};
\node (zd) at (1.1,-1.15) {${\scriptstyle \nu_n}$};
\end{tikzpicture}; x,y\right) +  \nonumber\\
&\qquad \qquad \qquad T \left( 
\begin{tikzpicture}[baseline=-27pt]
\node (c) at (0.5,-1.2) {}; %ellipse center
\draw [fill=gray!20] (c) ellipse (0.8cm and 0.4cm);
\node (aa) at (-0.2,-1.03) {$\bullet$};
\node (ab) at (0.1,-0.87) {$\bullet$};
\node (ac) at (0.4,-0.82) {$\bullet$};
\node (ad) at (1.2,-1.03) {$\bullet$};
\node [rotate=-10] (ae) at (0.85,-0.77) {$\cdots$};
\node (za) at (-0.1,-1.2) {${\scriptstyle{\nu_1}}$};
\node (zb) at (0.2,-1.05) {${\scriptstyle{\nu_2}}$};
\node (zc) at (0.55,-0.99) {${\scriptstyle \nu_3}$};
\node (zd) at (1.1,-1.15) {${\scriptstyle \nu_n}$};
\draw (aa.center) to [out=90, in=80, looseness=1.8] (ad.center); 
\draw (ab.center) to [out=80, in=90, looseness=1.5] (ad.center);
\draw (ac.center) to [out=80, in=100, looseness=1.3] (ad.center);
\end{tikzpicture}
;x,y\right) + \cdots + T \left( \begin{tikzpicture}[baseline=-32pt]
\node (c) at (0.5,-1.2) {}; %ellipse center
\draw [fill=gray!20] (c) ellipse (0.8cm and 0.4cm);
\node (aa) at (-0.2,-1.03) {$\bullet$};
\node (ab) at (0.1,-0.87) {$\bullet$};
\node (ac) at (0.4,-0.82) {$\bullet$};
\node (ad) at (1.2,-1.03) {$\bullet$};
\node [rotate=-10] (ae) at (0.85,-0.77) {$\cdots$};
\node (za) at (-0.1,-1.2) {${\scriptstyle{\nu_1}}$};
\node (zb) at (0.2,-1.05) {${\scriptstyle{\nu_2}}$};
\node (zc) at (0.55,-0.99) {${\scriptstyle \nu_3}$};
\node (zd) at (1.1,-1.15) {${\scriptstyle \nu_n}$};
\draw (aa.center) to [out=100, in=80, looseness=2.5] (ac.center); 
\draw (ab.center) to [out=110, in=90, looseness=2.5] (ac.center);
\end{tikzpicture}; x,y \right) + T \left( 
\begin{tikzpicture}[baseline=-35pt]
\node (c) at (0.5,-1.2) {}; %ellipse center
\draw [fill=gray!20] (c) ellipse (0.8cm and 0.4cm);
\node (aa) at (-0.2,-1.03) {$\bullet$};
\node (ab) at (0.1,-0.87) {$\bullet$};
\node (ac) at (0.4,-0.82) {$\bullet$};
\node (ad) at (1.2,-1.03) {$\bullet$};
\node [rotate=-10] (ae) at (0.85,-0.77) {$\cdots$};
\node (za) at (-0.1,-1.2) {${\scriptstyle{\nu_1}}$};
\node (zb) at (0.2,-1.05) {${\scriptstyle{\nu_2}}$};
\node (zc) at (0.55,-0.99) {${\scriptstyle \nu_3}$};
\node (zd) at (1.1,-1.15) {${\scriptstyle \nu_n}$};
\draw (aa.center) to [out=120, in=80, looseness=3] (ab.center); 
\end{tikzpicture}
;x,y\right)
\end{align}

{\noindent}where the grey blob represents the rest of the graph, with possible edges between the nodes labelled $\nu_1,\dots,\nu_n$. The identity is proven by repeatedly applying the deletion-contraction formula to the edges connected to $\mu$, starting with the edge connected to $\nu_n$. 
For traces in four-dimensions, there is a Chisholm identity for odd values of $n$ \cite{Chisholm:1972rw}:
\begin{equation}\label{eq:chisholm}
\mathrm{Tr}_4(\gamma_\mu \gamma_{\nu_1} \cdots \gamma_{\nu_n} \gamma_\mu) = -2\mathrm{Tr}_4(\gamma_{\nu_n} \cdots \gamma_{\nu_1})
\end{equation}
which translates into a theorem concerning the bicycle number:
\begin{theorem}\label{thm:odd}
For a graph $G$ with nodes $\mu$, $\nu_1,\dots,\nu_n$ where $n$ is odd and the only edges connected to $\mu$ are single edges to each of the nodes $\nu_1,\dots,\nu_n$, then there is a natural isomorphism between the bicycle space $B(G)$ of $G$, and the bicycle space $B((G - \mu) \star K_n)$ of a modified graph where the vertex $\mu$ has been removed, and edges have been attached between every pair of distinct vertices in the list $\nu_1,\dots,\nu_n$ (equivalently, a copy of the complete graph on $n$ nodes $K_n$ is glued onto the vertices $\{\nu_1,\dots,\nu_n\}$).  
\end{theorem}
\begin{equation}
B\left(\begin{tikzpicture}[baseline=-20pt]
\node (a) at (0.5,0) {$\bullet$};
\node (c) at (0.5,-1.2) {}; %ellipse center
\draw [fill=gray!20] (c) ellipse (0.8cm and 0.4cm);
\node (aa) at (-0.2,-1.03) {$\bullet$};
\node (ab) at (0.1,-0.87) {$\bullet$};
\node (ac) at (0.4,-0.82) {$\bullet$};
\node (ad) at (1.2,-1.03) {$\bullet$};
\node (ae) at (0.69,-0.6) {$\cdots$};
\node (za) at (-0.1,-1.2) {${\scriptstyle{\nu_1}}$};
\node (zb) at (0.2,-1.05) {${\scriptstyle{\nu_2}}$};
\node (zc) at (0.55,-0.99) {${\scriptstyle \nu_3}$};
\node (zd) at (1.1,-1.15) {${\scriptstyle \nu_n}$};
\node (ba) at (0.5,0.2) {$\mu$};
\draw (a.center) to (aa.center);
\draw (a.center) to (ab.center);
\draw (a.center) to (ac.center);
\draw (a.center) to (ad.center);
\end{tikzpicture}\right) \simeq B\left(
\begin{tikzpicture}[baseline=-27pt]
\node (c) at (0.5,-1.2) {}; %ellipse center
\draw [fill=gray!20] (c) ellipse (0.8cm and 0.4cm);
\node (aa) at (-0.2,-1.03) {$\bullet$};
\node (ab) at (0.1,-0.87) {$\bullet$};
\node (ac) at (0.4,-0.82) {$\bullet$};
\node (ad) at (1.2,-1.03) {$\bullet$};
\node [rotate=-10] (ae) at (0.85,-0.77) {$\cdots$};
\node (za) at (-0.1,-1.2) {${\scriptstyle{\nu_1}}$};
\node (zb) at (0.2,-1.05) {${\scriptstyle{\nu_2}}$};
\node (zc) at (0.55,-0.99) {${\scriptstyle \nu_3}$};
\node (zd) at (1.1,-1.15) {${\scriptstyle \nu_n}$};
\draw (aa.center) to [out=110, in=80, looseness=1.8] (ad.center); 
\draw (ab.center) to [out=100, in=90, looseness=1.5] (ad.center);
\draw (ac.center) to [out=80, in=100, looseness=1.3] (ad.center);
\draw (aa.center) to [out = 100, in = 90, looseness=1.8] (ac.center);
\draw (ab.center) to [out = 90, in = 90, looseness=1.5] (ac.center);
\draw (aa.center) to [out = 90, in = 90, looseness=1.5] (ab.center);
\end{tikzpicture}
\right)
\end{equation}
\begin{proof}
It is possible to utilise the Chisholm identity as well as \Cref{cor} to conclude at the very least that the dimensions of the respective bicycle spaces are the same. 
The following purely graph-theoretic proof presents more geometric intuition for the theorem, as well as providing an explicit isomorphism $f : B(G) \to B((G - \mu) \star K_n)$ between the two bicycle spaces. 

For any $X \in B(G)$, $X \in \mathrm{im}(\delta)$, so choose a representative $Y \in \delta^{-1}(X)$. The collection of all vertices in $Y$ that are not $\mu$ forms a subset of vertices $Y'$ in $(G - \mu) \star K_n$, define $f(X) := \delta(Y')$. 
For this map to be well defined, $\delta(Y')$ must first be invariant under the choice of $Y \in \delta^{-1}(X)$. 
To see this, first note that because $X$ is a bicycle, $X \in \mathrm{ker}(\partial)$ and hence there must be an even number of edges connected to $\mu$, corresponding to an even cardinality subset $N \subseteq \{\nu_1,\dots,\nu_n\}$ of neighbors of $\mu$ in the subgraph defined by $X$. 
If $\mu \notin Y$, then $Y \cap \{\nu_1,\dots,\nu_n\} = N$, conversely if $\mu \in Y$ then $Y \cap \{\nu_1,\dots,\nu_n\} = \{\nu_1,\dots,\nu_n\} - N$ owing to the fact that $\delta(Y) = X$. 
In either situation, $\delta(Y') \cap K_n$ is given by the complete bipartite graph $K_{|N|,n-|N|}$ which contains exactly every edge connecting vertices in $N$ to vertices outside of $N$ in the subset $\{\nu_1,\dots,\nu_n\}$. Hence, $\delta(Y')$ is invariant under your choice of $Y \in \delta^{-1}(X)$. 
Note that it is precisely because $n$ is odd that $\partial(\delta(Y')) = 0$, as the vertices in $N$ have an odd number of $K_n$ edges ($n - |N|$) connected to them. 

The above remarks have demonstrated a map $f : B(G) \to B((G-\mu) \star K_n)$. The inverse map can be constructed completely analogously, where for any $X \in B((G - \mu) \star K_n)$, pick a $Y \in \delta^{-1}(X)$, and set $f^{-1}(X) = \delta(Y)$ (where the vertices in $Y$ are now considered as vertices of $G$). 
\end{proof}

For even values of $n$ the corresponding statement about bicycle spaces does not hold (a counterexample is given by the four-cycle $G = C_4$, where $\mathrm{dim}(B(G)) = 1$ but $\mathrm{dim}(B((G - \mu) \star K_2)) = 0$ for any vertex $\mu \in C_4$). Instead, for even $n$ there is a different relation:
\begin{equation}\label{eq:chisholm2}
\gamma_\mu \gamma_{\nu_1} \cdots \gamma_{\nu_n} \gamma_\mu = 2 \gamma_{\nu_n} \gamma_{\nu_1} \cdots \gamma_{\nu_{n-1}} + 2 \gamma_{\nu_{n-1}} \cdots \gamma_{\nu_1} \gamma_{\nu_n}  
\end{equation}
which can be proven by applying the anticommutation relation to $\gamma_\mu$ on the right once, and then applying the Chisholm identity \Cref{eq:chisholm}. Naively, the number of terms generated by applications of the $d$-dimensional recurrence relation \Cref{eq:drel} and the $4$-dimensional recurrence relations \Cref{eq:chisholm,eq:chisholm2} in order to simplify a string of $\gamma$-matrices is exponential in the length of the string. A natural question to ask is what computational complexity classes evaluations of $d$-dimensional and $4$-dimensional traces belongs to, and this is discussed in \Cref{sec:comp}.

A simple reason that dedicated Tutte polynomial algorithms might be asymptotically faster than currently used $\mathrm{Tr}_d$ algorithms is that being able to use the Tutte polynomial deletion-contraction relation on graphs can relate Dirac traces to Tutte polynomial evaluations of graphs that are \textit{not} Circle graphs. To understand this in more detail, first we introduce the notion of local equivalence of graphs in order to understand the classification of Circle graphs. To locally complement a graph $G$ at a vertex $v$ is to replace the induced subgraph on all the neighbors of $v$ with the complement graph (to be precise, neighbors $w_i,w_j$ of $v$ are connected by an edge in the new graph if and only if they were not connected by an edge in the original graph) . If a graph $G_1$ can by a series of local complementations turn into a graph $G_2$ then they are said to be locally equivalent. 

\begin{theorem}[\cite{BOUCHET1994107}]\label{thm:fig2}
A simple graph $G$ is a circle graph if and only if no graph locally equivalent to $G$ has an induced subgraph isomorphic to one of the graphs depicted below:
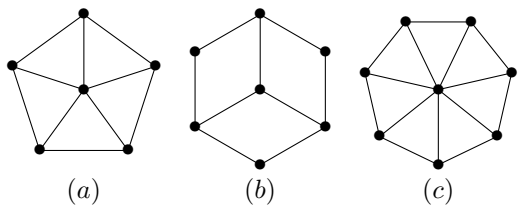
\begin{figure}[H]
\centering
\begin{tikzpicture}
\node (0) at (0,0) {$\bullet$};
\foreach \i in {1,...,5}
  {\node (\i) at ({cos(\i*360/5+90/5)},{sin(\i*360/5+90/5)}) {$\bullet$};
  \draw (0.center) to (\i.center);}
\foreach \i[evaluate={\j=int(mod(\i,5)+1;}] in {1,...,5}
  \draw (\i.center) to (\j.center);
  \node at (0,-1.35) {$(a)$};
\end{tikzpicture}
\begin{tikzpicture}
    \node (0) at (0,0) {$\bullet$};
\foreach \i in {1,...,6}
  \node (\i) at ({cos(\i*360/6+180/6)},{sin(\i*360/6+180/6)}) {$\bullet$};
\foreach \i in {1,3,5}
  \draw (0.center) to (\i.center);
\foreach \i[evaluate={\j=int(mod(\i,6)+1;}] in {1,...,6}
  \draw (\i.center) to (\j.center);
  \node at (0,-1.35) {$(b)$};
\end{tikzpicture}
\begin{tikzpicture}
\node (0) at (0,0) {$\bullet$};
\foreach \i in {1,...,7}
  {\node (\i) at ({cos(\i*360/7+90/7)},{sin(\i*360/7+90/7)}) {$\bullet$};
  \draw (0.center) to (\i.center);}
\foreach \i[evaluate={\j=int(mod(\i,7)+1;}] in {1,...,7}
  \draw (\i.center) to (\j.center);
  \node at (0,-1.35) {$(c)$};
\end{tikzpicture}
\caption{Key examples of graphs that are not circle graphs.}\label{fig2}
\end{figure}
\end{theorem}

With this classification of Circle graphs, it is possible to construct relationships between an evaluation of $\mathrm{Tr}_d$ and Tutte polynomial evaluations that \textit{cannot} be written as $d$-dimensional traces, for example in the evaluation of the trace of the following product of $12$ $\gamma$-matrices:
\begin{equation*}
\mathrm{Tr}_d(\gamma_{\mu_1} \gamma_{\mu_3} \gamma_{\mu_2} \gamma_{\mu_5} \gamma_{\mu_4} \gamma_{\mu_6} \gamma_{\mu_5} \gamma_{\mu_3} \gamma_{\mu_4} \gamma_{\mu_1} \gamma_{\mu_2} \gamma_{\mu_6}) =  128 d \cdot T \left(
\begin{tikzpicture}[baseline=0pt]
\node (0) at (0,0) {$\bullet$};
\foreach \i in {1,...,5}
  {\node (\i) at ({cos(\i*360/5+90/5)},{sin(\i*360/5+90/5)}) {$\bullet$};
  \node at ({1.3*cos(\i*360/5+90/5)},{1.3*sin(\i*360/5+90/5)}) {$\mu_{\i}$};
  \draw (0.center) to (\i.center);}
\foreach \i[evaluate={\j=int(mod(\i,5)+1;}] in {1,...,4}
  \draw (\i.center) to (\j.center);
  \node at (0.3,0.3) {$\mu_6$};
\end{tikzpicture} ;1-\frac{d}{2},-1\right)
\end{equation*}
\begin{equation}\label{eq:rec}
= 128d \left( T \left( \begin{tikzpicture}[baseline=0pt]
\node (0) at (0,0) {$\bullet$};
\foreach \i in {1,...,5}
  {\node (\i) at ({cos(\i*360/5+90/5)},{sin(\i*360/5+90/5)}) {$\bullet$};
  \draw (0.center) to (\i.center);}
\foreach \i[evaluate={\j=int(mod(\i,5)+1;}] in {1,...,5}
  \draw (\i.center) to (\j.center);
\end{tikzpicture}; 1-\frac{d}{2},-1\right)  - T \left( \begin{tikzpicture}[baseline=0pt]
\node (0) at (0,0) {$\bullet$};
\foreach \i in {2,...,4}
  {\node (\i) at ({cos(\i*360/5+90/5)},{sin(\i*360/5+90/5)}) {$\bullet$};
  \draw (0.center) to (\i.center);}
\node (15) at ({cos(180/5+90/5)},{sin(180/5+90/5)}) {$\bullet$};
\draw (2.center) to (15.center);
\draw (4.center) to (15.center);
\foreach \i[evaluate={\j=int(mod(\i,5)+1;}] in {2,...,3}
  \draw (\i.center) to (\j.center);
\draw (0.center) to [bend left=30] (15.center);
\draw (0.center) to [bend right=30] (15.center);
\end{tikzpicture}; 1-\frac{d}{2},-1\right) \right)
\end{equation}
derived by applying the deletion-contraction formula to the graph \Cref{fig2}(a). 
This formula however cannot be written as a relation on single dirac traces, as the first graph appearing in the final line of \Cref{eq:rec} is not a circle graph. 
Such relations open the door to potentially faster algorithms for computing $\mathrm{Tr}_d$. 
For example, a generic algorithm is to specify certain special classes of graphs where the answer is known and use the deletion-contraction relation to recursively travel towards these special cases.
All current algorithms for $\mathrm{Tr}_d$ rely on this idea where the special case are graphs with no edges, however (as described in Ref~\cite{Pemmaraju_Skiena_2003} for chromatic polynomial computations) a trivial improvement is to include the complete graphs as a special case as well. 
In the case of \Cref{eq:rec}, it would require adding $6$ edges to complete the graph, but require taking away $9$ edges to reduce it to the empty graph. 
For this idea to be useful, an explicit formula is required for the Tutte polynomial evaluated on the complete graph, which is provided below:

\begin{theorem}
\begin{align}
    \mathrm{Tr}_d\left(\gamma_{\mu_1} \cdots \gamma_{\mu_n} \gamma_{\mu_1} \cdots \gamma_{\mu_n} \right) &= 
    (-1)^{ 1 + \lceil\frac{n}{2}\rceil } 2^{1 + n} d \ T\left(K_n; 1 - \frac{d}{2},-1\right) \nonumber \\
    &= 4 \sum_{\lambda \vdash n} (-1)^{\sum_{i < j} \lambda_i \lambda_j} \frac{1}{c_1! \cdots c_k!} \left(\prod_{i=0}^{l-1} (d-i)\right) \frac{n!}{\lambda_1! \cdots \lambda_l!}
\end{align}
    where the sum is performed over all distinct partitions $\lambda = (\lambda_1,\dots,\lambda_l)$ of $n$ (such that $i < j$ implies $\lambda_i \leq \lambda_j$, and $\sum_{i=1}^l \lambda_i = n$), and $c_1,\dots,c_k$ are the multiplicities of the distinct numbers in the partition. 
\end{theorem}

\begin{proof}
    The proof follows from considering the representation $\mathrm{Tr}_d(x) = 4 (-1)^{|E|}  \ \chi(\mathrm{Gr}(x);-1,d)$ proven in \Cref{prop:a}. Each partition $\lambda = (\lambda_1,\dots,\lambda_l) \vdash n$ corresponds to a certain partition of the $n$ vertices into $l$ different colours, and the factor of $(-1)^{\sum_{i < j} \lambda_i \lambda_j}$ accounts for all the colour collisions. The factor of $\prod_{i=0}^{l-1} (d-i)$ counts the number of ways of assigning colours from $d$ choices to the different groups defined by the partition, but must be corrected by the factor of $\frac{1}{c_1! \cdots c_k!}$ to correct for overcounting. Finally, $\frac{n!}{\lambda_1!\cdots\lambda_l!}$ accounts for the different ways of distributing the colours among the vertices. 
\end{proof}

\section{Complexity of evaluating single Dirac traces} \label{sec:comp}

Given the connection developed between Dirac traces and computations of the Tutte polynomial, a natural question to consider is the asymptotic complexity of computing Dirac traces.
The algorithms to compute $\mathrm{Tr}_d$ with deletion-contraction relations as in \Cref{sec:trad} naively generate exponentially many terms, and one may wonder whether it is possible to find a polynomial time algorithm instead. 
As a quick review of the relevant complexity classes that will be discussed in this section, 
\begin{itemize}
    \item[-] $\P$ is the class of decision problems that can be solved by a deterministic Turing machine in a polynomial amount of computation time with respect to the size of the input. 
    \item[-] $\NP$ is the class of decision problems for which any successful query is accompanied by the existence of a polynomial-length `certificate', which can be verified in polynomial-time by a fixed deterministic algorithm. 
    \item[-] $\FP$ is the class of function problems which are computable in polynomial time.
    \item[-] $\#\P$ is the class of function problems corresponding to counting the number of certificates to given decision problems in $\NP$ \cite{VALIANT1979189}. 
    \item[-] $\P^{A}$ ($\FP^{A}$) for a given problem $A$ is the class of decision (function) problems which are computable in polynomial time given access to an oracle that either decides $A$ if $A$ is a decision problem, or computes $A$ if $A$ is a function problem. 
\end{itemize}
For more detailed definitions, see Ref~\cite{sanjeev}. 
This section will consider the problem of computing single Dirac traces with all indices contracted, and no $\gamma_5$ insertions.
Because these correspond to evaluations of the Tutte polynomial for circle graphs by \Cref{mainthm}, the computational problems we will be interested in are defined in terms of evaluations of the Tutte polynomial:
\begin{definition}
Let $\pi_X(\mathcal{C})$ be the function problem of computing the Tutte polynomial for a specified class of graphs $\mathcal{C}$ in the subset $X$ of the Tutte plane. $X$ can take three forms:
\begin{itemize}
    \item[-] $X$ is either a single point $X = \{(a,b)\}$ where $(a,b) \in \mathbb{R}^2$
    \item[-] $X$ is a rational curve, specified by a parametrisation $f : \mathbb{R} \to \mathbb{R}^2$ where $f(s) = \left(\frac{u(s)}{v(s)}, \frac{w(s)}{z(s)}\right)$ where $u,v,w,z$ are finite-degree polynomials in $s$ with coefficients in $\mathbb{Q}$. In this case, the Tutte polynomial output is considered as a rational function in the parameter $s$. 
    \item[-] $X$ is the entire plane $\mathbb{R}^2$, in which case the Tutte polynomial output is a polynomial in canonical variables $x,y$. 
\end{itemize}
\end{definition}
Computing single Dirac traces corresponds to computing the Tutte polynomial along the line $y = -1$ for circle graphs, and is written as $\pi_{\{(s,-1) : s \in \mathbb{R}\}}(\textbf{Circle})$. 
Mainly this section will deal with the question of whether or not $\pi_{\{(s,-1) : s \in \mathbb{R}\}}(\textbf{Circle}) \in \FP$, the class of functions which are computable in polynomial time, or whether $\pi_{\{(s,-1) : s \in \mathbb{R}\}}(\textbf{Circle})$ is contained in a complexity class believed to be distinct from $\FP$ (which would necessarily require superpolynomial time algorithms to compute), such as the class of $\#\P$-hard problems. 
Though it is possible that $\#\P=\FP$, it would for example immediately imply that $\P = \NP$ which has long been conjectured to be false \cite{cook}. 

It was shown that for the class of all graphs, $\pi_{\{(a,b)\}}(\mathbf{Graph}) \in \#\P$-hard, and $\pi_{\{f(s) : s \in \mathbb{R}\}}(\mathbf{Graph}) \\\in \#\P$-hard in Ref~\cite{Jaeger_Vertigan_Welsh_1990} for generic rational points $(a,b)$ and generic rational curves $f(s)$. These results were then extended to the class of planar graphs in Ref~\cite{vertigan}.  
The analysis for the complexity of evaluating Dirac traces is slightly complicated by the fact that Dirac traces are associated with evaluating Tutte polynomial along the specific curve $y = -1$, and for the class of circle graphs. 
It is however quite natural to conjecture that Dirac traces would also be $\#\P$-hard in the absence of some special structures present only within the class of circle graphs. 
Towards the goal of proving this statement, we follow the strategy used in Ref~\cite{Jaeger_Vertigan_Welsh_1990} for the case of Tutte polynomial evaluations for the class of all graphs. 
The first objective is to identify specific points within the Tutte plane for which evaluations are difficult. 
For the class of circle graphs, deciding whether a given circle graph has a proper $n$-colouring for $n \geq 4$ is known to be $\NP$-complete by reduction to 3-SAT:

\begin{theorem}[\cite{Unger}, Thm. 1]\label{thm:ung}
The problem of deciding whether a given circle graph is $n$-colourable is $\NP$-complete for $n \geq 4$. 
\end{theorem}

Unfortunately, this does not immediately imply that the corresponding function problem of computing $T(G;1-n,0)$ for $n \geq 4$ and circle graphs $G$ is $\#\P$-complete, as this would require finding a parsimonious reduction (one that preserves the number of solutions) from $n$-colourability to 3-SAT. 
It may be possible to modify the reduction in Ref~\cite{Unger} to become parsimonious, for now, we will be satisfied with the following corollary:
\begin{cor}\label{cor:NP}
$\pi_{\{(1-n,0)\}}({\normalfont{\textbf{Circle}}})$ for integer $n \geq 4$ (and hence, also $\pi_{\mathbb{R}^2}({\normalfont{\textbf{Circle}}})$) are $\NP$-hard function problems. 
\end{cor}
\begin{proof}
The only technical remark to make is that because $\pi_{\{(1-n,0)\}}({\normalfont{\textbf{Circle}}})$ is a function problem, what is meant by $\NP$-hard in this context is not a polynomial time many-to-one reduction of decision problems (as in Ref~\cite{Karp1972}), but rather that given access to a $\pi_{\{(1-n,0)\}}({\normalfont{\textbf{Circle}}})$-oracle, one can solve any $\NP$ problem in polynomial time, $\NP \subseteq P^{\pi_{\{(1-n,0)\}}({\normalfont{\textbf{Circle}}})}$. 
\Cref{cor:NP} then follows from \Cref{thm:ung}, because deciding whether a given circle graph $G$ is $n$-colourable can be done by computing the chromatic polynomial $\chi(G;0,n) = n^{c(G)} (-1)^{|V|-c(G)} T\left(G; 1 -n, 0\right)$ and checking whether or not the result equals zero. 

This immediately implies $\pi_{\mathbb{R}^2}(\textbf{Circle})$ is also $\NP$-hard, because being able to compute the Tutte polynomial on the entire plane allows you in polynomial time to evaluate the Tutte polynomial at any specific point $(1-n,0)$  by evaluating the polynomial at $x = 1-n, y = 0$. 
\end{proof}
The strategy proceeds by utilising a general theorem proven in Ref~\cite{Jaeger_Vertigan_Welsh_1990} that $\pi_{\mathbb{R}^2}(\mathcal{C})$ is polynomial time reducible to $\pi_{L}(\mathbb{C})$ for generic rational curves $L$, and as long as the class of graphs $\mathcal{C}$ is `large' enough in a technical sense. 
The theorems quoted below from Ref~\cite{Jaeger_Vertigan_Welsh_1990} are all stated in the more general case of matroids, which we now provide a brief a review of. 
A matroid is a pair $M = (S,\rho)$ of a ground set $S$ and a rank function $\rho : 2^S \to \mathbb{Z}$ satisfying $0 \leq \rho(A) \leq |A|$ for all $A \subseteq S$, $\rho(X) \leq \rho(Y)$ for $X \subseteq Y \subseteq S$, and $\rho(X \cup Y) + \rho(X \cap Y) \leq \rho(X) + \rho(Y)$. 
Every graph $G = (V,E)$ gives rise to a `graphic' matroid by letting the ground set be the set of edges $E$, and letting the rank function $\rho(X) = n - c$ where $n$ is the number of vertices in the subgraph formed by the edges in $X$ and $c$ is the number of connected components of the subgraph. 
Instead of specifying the entire rank function $\rho$, a matroid can also be specified by its bases, which are the maximal independent subsets $X \subset S$, where an independent subset $X$ is one that satisfies $\rho(X) = |X|$. 
In the case of graphic matroids, $\rho(X) = |X|$ if and only if $X$ has no cycles ($X$ is a forest). 

The number $f(n)$ of non-isomorphic matroids on a set of $n$ elements grows at least as fast as $\log_2(\log_2(f(n))) \geq n-\frac{3}{2} \log_2(n) + O(\log \log n)$\cite{KNUTH1974398}. As pointed out in Ref~\cite{Jaeger_Vertigan_Welsh_1990}, this implies both the input size, and the runtime for computing the Tutte polynomial on general matroids is exponential in $n$, as the naive algorithm utilising deletion-recurrence relations runs in time exponential in $n$.
The class of matroids corresponding to circle graphs can be represented more succinctly however:

\begin{definition}[\cite{Jaeger_Vertigan_Welsh_1990}]
A class of matroids $\mathcal{C}$ is succinct if there is an injective mapping $e$ of the members of $\mathcal{C}$ into strings of some finite alphabet $\Sigma$, such that if $|e(M)|$ denotes the length of $e(M)$
\begin{equation}
\mathcal{C}_n := \{ (S,\rho) \in \mathcal{C} : |S| = n\}, \quad e(n) := \mathrm{max}\{|e(M)| : M \in \mathcal{C}_n\},
\end{equation}
then there exists some polynomial $p$ such that $e(n) \leq p(n)$ for all $n$. 
\end{definition}

\begin{lem}
The class of circle graphs is succinct.
\end{lem}

\begin{proof}
Succinctness of the class of circle graphs follows from succinctness of the class of all simple graphs. 
An explicit succinct encoding of simple graphs is given by writing down the adjacency matrix of a simple graph. 
\end{proof}

A class of simple transformations on matroids are given by tensor products with the uniform matroid $U_{r,n}$, which has ground set of $n$ elements and bases all subsets of size $r$. 
To define this tensor product, a few background definitions are required. 
A pointed matroid $N_d$ is a matroid on a ground set with a distinguished element, the point $d$, which is neither a loop ($\rho(\{d\}) = 0$) nor a coloop ($\rho(S - \{d\}) = \rho(S) - 1$). 
If $M = (S,\rho), N = (T,\lambda)$ are two matroids with $e \in S, f \in T$, the $2$-sum of $M$ and $N$ is the matroid $((S \cup T) \backslash \{e,f\} , \sigma)$ such that for $A \subseteq S \backslash \{e\}, B \subseteq T \backslash \{f\}$ 
\begin{equation}
\sigma(A \cup B) = \rho(A) + \lambda(B) - \delta(A,B) + \delta(\emptyset,\emptyset)
\end{equation}
where $\delta(A,B) = 1$ if both $\rho(A) = \rho(A \cup \{e\})$ and $\lambda(B) = \lambda(B \cup \{f\})$, else $\delta(A,B) = 0$. 
The tensor product $M \otimes N$ of an arbitrary matroid $M$ and the pointed matroid $N$ is obtained by taking the $2$-sum of $M$ with $N$ at each point $e$ of $M$ and the distinguished point $d$ of $N$ \cite{Brylawski}. 
For tensor products with the uniform matroid, it does not matter which element is chosen as the distinguished point, as the matroid is invariant under permutations of the base set. 

Not all uniform matroids are graphic, however the cases of interest are. In particular, $U_{k,k+1}$ corresponds to $(k+1)$-edged cycle graph, and $U_{1,k+1}$ corresponds to the $(k+1)$-edged dipole graph ($(k+1)$ edges connecting two vertices). The tensor product of a graphic matroid with $U_{k,k+1}$ has the effect of replacing every edge with a $k$-stretched version, and the tensor product of a graphic matroid with $U_{1,k+1}$ with a $k$-thickened version, hence the naming convention. For example, $2$-stretchings and $2$-thickenings of $K_4$ are shown below:
\begin{align}
\begin{tikzpicture}[baseline=0pt]
\def\length{0.75}
\node (0) at (0,0) {$\bullet$};
\foreach \i in {1,2,3}
  \node (\i) at ({\length*cos((\i-1)*120+90)},{\length*sin((\i-1)*120+90}) {$\bullet$};
\foreach \i in {0,...,3}
  \foreach \j in {0,...,3}
    \draw (\i.center) to (\j.center);
\end{tikzpicture}
\otimes 
\begin{tikzpicture}[baseline=0pt]
\def\length{0.75}
\foreach \i in {0,...,2}
  \node (\i) at ({\length*cos(\i*120+90)},{\length*sin(\i*120+90)}) {$\bullet$};
\foreach \i[evaluate={\j=int(mod(\i+1,3);}] in {0,...,2}
  \draw (\i.center) to (\j.center);
\end{tikzpicture}
 = 
 \begin{tikzpicture}[baseline=0pt]
\def\length{0.75}
\node (0) at (0,0) {$\bullet$};
\foreach \i in {1,2,3}
  \node (\i) at ({\length*cos((\i-1)*120+90)},{\length*sin((\i-1)*120+90}) {$\bullet$};
\foreach \i in {0,...,3}
  \foreach \j in {0,...,3}
    {\node at ($(\i)+.5*(\j)-.5*(\i)$) {$\bullet$};
    \draw (\i.center) to (\j.center);}
 \end{tikzpicture},\qquad 
 \begin{tikzpicture}[baseline=0pt]
\def\length{0.75}
\node (0) at (0,0) {$\bullet$};
\foreach \i in {1,2,3}
  \node (\i) at ({\length*cos((\i-1)*120+90)},{\length*sin((\i-1)*120+90}) {$\bullet$};
\foreach \i in {0,...,3}
  \foreach \j in {0,...,3}
    \draw (\i.center) to (\j.center);
\end{tikzpicture}
\otimes 
\begin{tikzpicture}[baseline=0pt]
\def\length{0.75}
\node (0) at (0,0.7) {$\bullet$};
\node (1) at (0,-0.3) {$\bullet$};
\foreach \i in {0,...,2}
  \draw (0.center) to [bend right=-40+\i*40] (1.center);
\end{tikzpicture}
 = 
\begin{tikzpicture}[baseline=0pt]
\def\length{0.75}
\node (0) at (0,0) {$\bullet$};
\foreach \i in {1,2,3}
  \node (\i) at ({\length*cos((\i-1)*120+90)},{\length*sin((\i-1)*120+90}) {$\bullet$};
\foreach \i in {0,...,3}
  \foreach \j in {0,...,3}
    \foreach \k in {-1,1}
      \draw (\i.center) to [bend right=\k*10] (\j.center);
\end{tikzpicture} \\
K_4 \hspace{1.75cm} U_{2,3} \hspace{4.2cm}K_4 \hspace{1cm}U_{1,3}\hspace{2cm} \nonumber 
\end{align}

The utility of thinking about these tensor products is that evaluating the Tutte polynomial on the $k$-lengthened or $k$-thickened versions of a matroid at a specific point $(x,y)$ yields the Tutte polynomial for the original matroid at a different point $(x',y')$. Specifically,  

\begin{theorem}[\cite{Jaeger_Vertigan_Welsh_1990}, Eq. 4.1]
For a matroid $M = (S,\rho)$ and a pointed matroid $N$, 
    \begin{equation}
T(M \otimes N; x,y) = T_C(N;x,y)^{|S| - \rho(S)} T_L(N; x,y)^{\rho(S)} T(M; X,Y)
    \end{equation}
where for $k$-lengthening:
\begin{equation}
T_C(U_{k,k+1};x,y) = \sum_{i=0}^{k-1} x^i, \quad T_L(U_{k,k+1};x,y) = 1, \quad X = x^k, \quad Y = \frac{x^k + x(y-1) - y}{x^k - 1}\end{equation}
and for $k$-thickening:
\begin{equation}\label{eq:mod}
T_C(U_{1,k+1};x,y) = 1, \quad T_L(U_{k,k+1};x,y) = \sum_{i=0}^{k-1} y^i, \quad X = \frac{y^k + x(y-1) - y}{y^k - 1}, \quad Y = y^k
\end{equation}
\end{theorem}
Thus, being able to evaluate the Tutte polynomial efficiently for a class of graphs closed under $k$-thickening and $k$-lengthening allows one to take a given graph, and evaluate it's Tutte polynomial on the transformed $(X,Y)$ coordinates, which in turn allows for an interpolation of the Tutte polynomial on the entire plane provided the curve is generic. Note that the $k$-thickening operation is not useful for the problem $\pi_{\{(s,-1) : s \in \mathbb{R}\}}(\textbf{Circle})$ because the class of circle graphs contains only simple graphs and does not contain any $k$-thickened graphs. Furthermore by inspecting \Cref{eq:mod}, it is clear that the $k$-thickening transformation $(x,y) \mapsto(X,Y)$ is only well-defined for the line defined by $y = -1$ by taking some limit as $y \to -1$, complicating the situation. By focusing only on the $k$-lengthening operation, the main theorem of Ref~\cite{Jaeger_Vertigan_Welsh_1990} is slightly modified:

\begin{theorem}[\cite{Jaeger_Vertigan_Welsh_1990}, Thm. 1 (modified)]\label{thm:jmod}
Let $L$ be a rational curve $f(s) = (x(s),y(s)) = \left( \frac{u(s)}{v(s)},\frac{w(s)}{z(s)}\right)$ for polynomials $u,v,w,z$ with coefficients in $\mathbb{Q}$ such that $(x(s)-1)(y(s)-1)$ is not constant, and such that $x(s)$ is also not constant. 
Let $\mathcal{C}$ be a succinct class of matroids closed under $k$-stretching, where computing the $k$-stretch of a succinct representation of a matroid takes polynomial time both in $k$ and the size of the matroid, yielding a succinct representation of the stretched matroid. 
Then $\pi_{\mathbb{R}^2}(\mathcal{C})$ is polynomial time reducible to $\pi_{L}(\mathcal{C})$. 
\end{theorem}

{\noindent}Unfortunately, the class of circle graphs is not closed under $k$-lengthenings. 

\begin{prop}
The class of circle graphs is not closed under lengthenings. 
\end{prop}

\begin{proof}
An explicit example of a graph $G$ whose $2$-lengthening is not a circle graph is provided below:
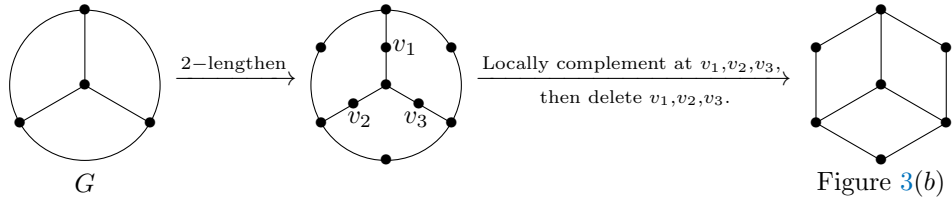
\begin{figure}[H]
    \centering
    \begin{tikzpicture}[baseline=0pt]
\node (0) at (0,0) {$\bullet$};
\foreach \i in {1,3,5}
  \node (\i) at ({cos(\i*360/6+180/6)},{sin(\i*360/6+180/6)}) {$\bullet$};
\foreach \i in {1,3,5}
  \draw (0.center) to (\i.center);
\draw (0,0) circle (1);
\node at (0,-1.3) {$G$};
    \end{tikzpicture}
    $\xrightarrow{2-\mathrm{lengthen}}$
    \begin{tikzpicture}[baseline=0pt]
\node (0) at (0,0) {$\bullet$};
\foreach \i in {1,...,6}
  \node (\i) at ({cos(\i*360/6+180/6)},{sin(\i*360/6+180/6)}) {$\bullet$};
\foreach \i in {1,3,5}
  {\draw (0.center) to (\i.center);
  \node (\i) at ({.5*cos(\i*360/6+180/6)},{.5*sin(\i*360/6+180/6)}) {$\bullet$};}
  \node at (0.25,0.5) {$v_1$};
  \node at (-0.35,-0.45) {$v_2$};
  \node at (0.4,-0.45) {$v_3$};
\draw (0,0) circle (1);
    \end{tikzpicture}
    $\xrightarrow[{ \text{ then delete } v_1,v_2,v_3.}]{{\text{Locally complement at } v_1,v_2,v_3,}}$
    \begin{tikzpicture}[baseline=0pt]
    \node (0) at (0,0) {$\bullet$};
\foreach \i in {1,...,6}
  \node (\i) at ({cos(\i*360/6+180/6)},{sin(\i*360/6+180/6)}) {$\bullet$};
\foreach \i in {1,3,5}
  \draw (0.center) to (\i.center);
\foreach \i[evaluate={\j=int(mod(\i,6)+1;}] in {1,...,6}
  \draw (\i.center) to (\j.center);
  \node at (0,-1.3) {\Cref{fig2}$(b)$};
\end{tikzpicture}
\caption{Example of a circle graph $G$ whose $2$-lengthening is not a circle graph.}\label{fig:ex}
\end{figure}
{\noindent}But $G$ is a circle graph, since it has $4$ vertices, and hence cannot be locally equivalent to a graph that contains one of the graphs in \Cref{thm:fig2} as an induced subgraph, as they all have at least $6$ vertices. An explicit example is given by $\mathrm{Gr}((1,2,3,4,1,2,3,4))$. 
\end{proof}

A weaker statement is of course possible to make:
\begin{theorem}\label{thm:cstar}
Let the class $\mathrm{Circle}^*$ be the class of circle graphs and all $k$-lengthenings of circle graphs. Then computing $\mathrm{Tr}_d$ on this class is $\NP$-hard. 
\end{theorem}
\begin{proof}
Follows from \Cref{thm:jmod,cor:NP}. 
\end{proof}

To circumvent the issue of circle graphs not being closed under $k$-stretchings, it is possible to imagine that there is a subset of circle graphs closed under stretchings yet where determining whether the graph is $k$-colourable is still NP-hard. A candidate for such a class is the class of circle graphs that do not contain a $4$-clique, as Ref~\cite{Unger} comments that the $k$-colourability problem is still NP-hard on this class. Unfortunately, this class is also not closed under lengthenings, for example the graph $G'$ obtained from $G$ in \Cref{fig:ex} by taking a single edge and subdividing it into two pieces, is a circle graph with no $4$-clique but whose $2$-lengthening is not a circle graph. 

The situation is considerably simpler for $4$-dimensional traces, for which the computational problem $\pi_{\{(-1,-1)\}}(\mathbf{Circle})$ corresponds to evaluating the Tutte polynomial at the special point $(-1,-1)$.

\begin{theorem}\label{thm:4P}
Computing $4$-dimensional traces is in $\FP$. 
\end{theorem}

\begin{proof}
Let $G = (V,E)$ denote a circle graph. 
Note that one choice of succinct representation consists of writing down its adjacency matrix, whose size scales as $O(|V|)^2$. 
The bicycle space $\mathrm{im}(\delta) \cap \mathrm{ker}(\partial)$ can be computed in polynomial time by gaussian elimination over the field $\mathbb{Z}_2$, as pointed out in Ref~\cite{Jaeger_Vertigan_Welsh_1990}. 
Computing the 4-dimensional trace by utilising $\mathrm{Tr}_4(x) = 4(-2)^{|V| + c(\mathrm{Gr(x)}) + \mathrm{dim}(B(\mathrm{Gr}(x)))}$ is then a polynomial in $|V|$ overhead, as the number of connected components of $\mathrm{Gr}(x)$ can be computed with breadth/depth-first search. 
\end{proof}

In practice, dimensional regularisation calculations usually only require expansions of $\mathrm{Tr}_d$ around $d = 4$ at some finite fixed order (corresponding to the number of loops in the diagram). 
It may be the case that these expansions also have a graph-theoretic interpretation, and can be calculated in polynomial time. 
In light of these partial results, a complete understanding of the computational complexity of computing the Tutte polynomial for the class of circle graphs would be an interesting study to see. 

\section{Conclusion}

In this work, a connection between dimensionally regulated Dirac traces and Tutte polynomials of corresponding graphs was proven. The connection is theoretically interesting in both directions, on the one hand giving a graph-theoretic interpretation of what dimensionally regulated traces means, and on the other hand providing a `physical' interpretation for the Tutte polynomial $T(G;x,y)$ evaluated along $y = -1$ (at least when $G$ is a circle graph). 
Preliminary investigations suggest the connection may also be practically useful, providing faster algorithms to compute Dirac traces for single traces of randomly contracted Dirac matrices. 
These preliminary investigations have to be taken with a grain of salt, given that the distribution of Dirac traces that actually appear in realistic multi-loop calculations is not going to be the uniformly random contracted distribution. 
It would however be interesting to see in the future a complete implementation of arbitrary Dirac traces utilising the connection to Tutte polynomials developed in this work. 

The discussion in \Cref{sec:comp} motivates one to consider a more in-depth study of the computational complexity of Tutte polynomial evaluations on the class of circle graphs (similar to what was performed for planar graphs in Ref~\cite{vertigan}). 
For instance, a natural conjecture to consider is that \Cref{thm:cstar} can be strengthened to a statement that evaluations of $d$-dimensional Dirac traces is $\#\P$-hard. 
Following the strategy outlined in \Cref{sec:comp}, this would first require showing that $\pi_{\{(1-n,0)\}}(\mathbf{Circle})$ are $\#\P$-complete, possibly by showing that the reduction in Ref~\cite{Unger} can be made parsimonious. 
Then perhaps by finding a different transformation (other than $k$-stretching) that circle graphs are closed under, that allows one to interpolate the Tutte polynomial on the entire plane. 
In a similar vein, it would be interesting to consider the computational problem of computing the coefficient of $\epsilon^n$ in $\mathrm{Tr}_d$, where $d = 4 - \epsilon$ for a fixed $n$ (as in practical perturbative QFT calculations calculations are only performed to some fixed number of loops). 

Though a framework for handling open indices, multiple traces and the `tHooft-Veltman $\gamma_5$ is presented in \Cref{sec:4}, it's possible that there is a more natural framework for these objects as well. A more matroid-theoretic interpretation of products of multiple Dirac traces with indices contracted in some fashion may allow for a clean analysis of the computational complexity of this more complicated case. Graph-theoretic interpretations of $\gamma_5$, evanescent operators, and similar complications in dimensionally regulated Dirac matrices may prove fruitful in the future. 

\section*{Acknowledgement}
The author gratefully acknowledges useful discussions with Ryan Abbott and William Detmold. The author is supported in part by the U.S. Department of Energy, Office of Science under grant Contract Number
DE-SC0011090 and by the SciDAC5 award DE-SC0023116, and is additionally supported by the National Science Foundation under Cooperative Agreement PHY-2019786 (The NSF AI Institute for Artificial Intelligence and Fundamental Interactions, http://iaifi.org/).

\begin{appendix}

\section{Open indices, multiple traces and \texorpdfstring{$\gamma_5$}{gamma5}}
\label{sec:4}

The formalism established in \Cref{sec:proof} deals with the case where all indices within a single Dirac trace are contracted amongst each other, and there are no $\gamma_5$ insertions. 
In practical calculations, these conditions are violated if there are multiple fermion loops connected by some interactions, or the theory under investigation contains chiral interactions. 
Formally speaking, the problem of computing products of Dirac traces with indices contracted in some arbitrary fashion can be reduced to the problem of computing single Dirac traces with some open (uncontracted) indices. Furthermore, Dirac traces with $\gamma_5$ insertions (in the `t-Hooft Veltman $\gamma_5$ scheme) can similarly be reduced to the problem of computing Dirac traces with open indices.
Though it is standard knowledge about how to reduce the problem of traces with open indices to traces with all indices contracted, it seems helpful to provide a perspective on how to formalise this procedure. 

To formalise the problem, it is easier to now specify a specific alphabet of symbols. Let $P(n,m)$ be the set of all tuples of length $2n+2m$ where each symbol in $\{\mu_1,\dots,\mu_n\}$ appears twice (these are the contracted indices) and each symbol in $\{\nu_1,\dots,\nu_{2m}\}$ appears only once. Unlike in the case with all indices contracted, recursion relations of the form \Cref{eq:rec1,eq:rec2} are not sufficient to determine $\mathrm{Tr}_d$ on $P(n,m)$ due to the fact that anticommuting uncontracted indices doesn't help simplify the expression. The additional physical axiom that is imposed is that the only structures that can appear in the trace $\mathrm{Tr}_d(x)$ for $x \in P(n,m)$ are metric tensors $g_{\nu_i,\nu_j}$. Let $Q(m)$ be the set of perfect matchings on the complete graph $K_{2m}$ with labelled vertices $\{1,\dots,2m\}$. Then the trace can be written as a linear combination:
\begin{equation}
\mathrm{Tr}_d(x) = \sum_{G \in Q(m)} C_G \left( \prod_{e \in G} g_{\nu_{\mathrm{src}(e)} \nu_{\mathrm{snk}(e)}} \right)
\end{equation}
for some constants $C_G$, where for each edge $e$ in each perfect matching $G$, an arbitrary orientation has been picked. This arbitrary orientation has no affect on the resulting expression, as the metric tensor is symmetric ($g_{\alpha \beta} = g_{\beta \alpha}$). The problem has thus been reduced to determining the constants $C_G$. What is calculable with the technology introduced thus far are traces where the uncontracted indices are contracted, in other words for some $G' \in Q(m)$, 
\begin{align}\label{eq:ui}
\mathrm{Tr}_d\left(\prod_{e \in G'} g_{\nu_{\mathrm{src}(e)} \nu_{\mathrm{snk}(e)}}  \cdot x\right) &= \sum_{G \in Q(m)} C_{G}   \left(\prod_{e' \in G'} g_{\nu_{\mathrm{src}(e')} \nu_{\mathrm{snk}(e')}} \right)\left(\prod_{e \in G}  g_{\nu_{\mathrm{src}(e)} \nu_{\mathrm{snk}(e)}}\right)   \nonumber \\
&= \sum_{G \in Q(m)} C_{G}\  d^{\text{ cyc}(G \star G')}
\end{align}
where `$\mathrm{cyc}$' is the function that counts the number of undirected cycles of a graph, and $G \star G'$ is the graph obtained by gluing the labelled vertices of the two graphs together. The left-hand-side of \Cref{eq:ui} has no uncontracted indices, and can be computed by a Tutte polynomial evaluation of the corresponding graph, and the right hand side is a linear combination of the unknown $C_{G}$ coefficients. By computing \Cref{eq:ui} for all $G' \in Q(m)$, enough contraints are found to constrain $C_G$ for all $G \in Q(m)$. With the observations made above, the Dirac trace with open indices can be defined as follows:

\begin{definition}\label{defn:uncontracted}
Let $A$ be a $|Q(m)| \times |Q(m)|$ sized matrix, whose entries are given by:
\begin{equation}
A_{G_1,G_2} = d^{\ \mathrm{ cyc}(G_1 \star G_2)}
\end{equation}
The $d$-dimensional trace is extended to $P(n,m)$ by defining:
\begin{equation}
\mathrm{Tr}_d(x) = \sum_{G_1,G_2 \in Q(m)} 
 \mathrm{Tr}_d \left( \prod_{e \in G_1} g_{\nu_{\mathrm{src}(e)}\nu_{\mathrm{snk}(e)}} x \right) (A^{-1})_{G_1,G_2} \left(\prod_{e \in G_2} g_{\nu_{\mathrm{src}(e)} \nu_{\mathrm{snk}(e)}} \right)
\end{equation}
\end{definition}

A practical way to approach the problem is to use anticommutation relations to move all the uncontracted indices to the left of the string. The contracted indices then can be handled via a Tutte polynomial evaluation, and what's left is a trace over purely uncontracted indices $\mathrm{Tr}_d(\nu_1, \nu_2,\dots,\nu_{2m})$. 
There is a simple formula for this trace, which requires first a few technical definitions and statements:

\begin{definition}
The sign $\mathrm{sign}(G)$ of a perfect matching $G \in Q(n)$ on the complete graph $K_{2n}$ with labelled vertices $\{1,\dots,2n\}$ is given by the sign of the naturally constructed involution in the symmetric group $S_{2n}$. Explicitly:
\begin{equation}
\mathrm{sign}(G) := \prod_{e \in G} (-1)^{\mathrm{snk}(e) - \mathrm{src}(e) + 1}
\end{equation}
where an arbitrary orientation has been chosen for the edges $e \in G$. 
\end{definition}

\begin{lem}
For $n \geq 0$, 
\begin{equation}
    \sum_{G \in Q(n)} \mathrm{sign}(G) = 1
\end{equation}
\end{lem}
\begin{proof}
This lemma can be proven by induction on $n$. For the basecase, $Q(0)$ contains only the empty perfect matching, which has sign $1$. For any $n \geq 1$, the node labelled by $1$ is matched to some other node $i \in \{2,\dots,2n\}$. The result of summing over the signs of all such perfect matchings fixing $i$ gives $(-1)^i \sum_{G \in Q(n-1)} \mathrm{sign}(G) = (-1)^i$ by inductive hypothesis. Finally, summing over $i \in \{2,\dots,2n\}$ proves the lemma.  
\end{proof}

As an intermediate step, the following theorem gives an interesting alternate way of evaluating the $d$-dimensional trace:

\begin{theorem}\label{thm:interesting}
For $x \in P(n)$, 
\begin{equation}\label{target}
\mathrm{Tr}_d (x) = 4 \sum_{G \in Q(n)} (-1)^{\mathrm{sign}(G)} d^{\ \mathrm{cyc}(G \star \mathrm{Gr}'(x))} 
\end{equation}
where $\mathrm{Gr}'(x)$ is the graph on $2n$ vertices where vertex $i$ is connected to vertex $j$ by an edge if and only if $x_i = x_j$. 
\end{theorem}

\begin{proof}
Fix a positive integer value of $d$, and consider the generalised chromatic number formula for $\mathrm{Tr}_d$ given in \Cref{prop:a}:
\begin{equation}\label{temp}
\mathrm{Tr}_d(x) = 4 (-1)^{|E|}  \ \chi(\mathrm{Gr}(x);-1,d)
\end{equation}
The RHS of \Cref{target} has a similar structure to the RHS of \Cref{temp}. For every graph $G \in Q(n)$, $d^{\ \mathrm{cyc}(G \star \mathrm{Gr}(x))}$ is counting the number of non-proper vertex-colourings of $\mathrm{Gr'}(x)$ if there are $d$ colours where each cycle is coloured the same. As the cycles are colored the same, each of these colorings can be associated with a colouring appearing in $\chi(\mathrm{Gr}(x); -1,d)$. For a single term in the LHS sum (where each of the contractions is colored with a specific color), it matches correctly onto the contributions on the right, by some kind of inclusion-exclusion. 
\end{proof}

Finally, traces with all uncontracted indices can be evaluated by the following formula:
\begin{theorem}
For $x = (\nu_1,\cdots,\nu_{2m}) \in P(0,m)$, the equation:
\begin{equation}\label{eq:perm}
\mathrm{Tr}_d(x) = 4 \sum_{G \in Q(m)} (-1)^{\mathrm{sign}(G)} \left( \prod_{e \in G} g_{\nu_{\mathrm{src}(e)} \nu_{\mathrm{snk}(e)}} \right) 
\end{equation}
holds where $\mathrm{sign}(G)$ indicates its sign when canonically associated to a permutation in $S_{2m}$. 
\end{theorem}
\begin{proof}
Follows from \Cref{defn:uncontracted,thm:interesting}. 
\end{proof}

For example for $n\in\{1,2\}$ \Cref{eq:perm} takes the form:
\begin{align}
\mathrm{Tr}_d(\gamma_{\nu_1} \gamma_{\nu_2}) &= 4 g_{\nu_1\nu_2} \nonumber \\
\mathrm{Tr}_d(\gamma_{\nu_1}\gamma_{\nu_2} \gamma_{\nu_3}\gamma_{\nu_4}) &= 4 g_{\nu_1 \nu_2} g_{\nu_3 \nu_4} - 4 g_{\nu_1 \nu_3} g_{\nu_2 \nu_4} + 4 g_{\nu_1 \nu_4} g_{\nu_2 \nu_3} 
\end{align}
Now that the technology is set up to treat traces including open indices, extensions to multiple traces and the `t-Hooft-Veltman (HV) $\gamma_5$ is not difficult. In particular, products of multiple traces where indices are possible contracted amongst the different traces can formally be handled by doing single dirac traces with open indices, and contracting the resulting metric tensors. Also, the HV $\gamma_5$ is specified by:
\begin{equation}
\gamma_5 := \frac{1}{4!} \tilde{\epsilon}_{\mu_1 \mu_2 \mu_3 \mu_4} \gamma_{\mu_1} \gamma_{\mu_2} \gamma_{\mu_3} \gamma_{\mu_4} 
\end{equation}
where $\tilde{\epsilon}$ is the fully-antisymmetric tensor in four dimensions; hence traces involving $\gamma_5$ are performed simply as traces with open indices. The $\sim$ emphasizes that the $\epsilon$-tensor is a 4-dimensional object in the HV-scheme. The 4-dimensional tensors satisfy the contraction relations:
\begin{equation}
\tilde{\epsilon}_{\mu_1 \mu_2 \mu_3 \mu_4} \tilde{\epsilon}_{\nu_1 \nu_2 \nu_3 \nu_4} = \sum_{\sigma \in S_4} \mathrm{sign}(\sigma) \prod_{i = 1}^4 \tilde{g}_{\mu_i \nu_{\sigma(i)}},
\end{equation}
\begin{equation}
g_{\nu \mu_1} \tilde{\epsilon}_{\mu_1 \mu_2 \mu_3 \mu_4} = \tilde{g}_{\nu \mu_1} \tilde{\epsilon}_{\mu_1 \mu_2 \mu_3 \mu_4} = \tilde{\epsilon}_{\nu \mu_2 \mu_3 \mu_4}, \qquad  
\tilde{g}_{\nu \mu_1} g_{\nu \mu_2} = \tilde{g}_{\nu \mu_1} \tilde{g}_{\nu \mu_2} = \tilde{g}_{\mu_1 \mu_2}, \qquad \tilde{g}_{\mu \mu} = 4.
\end{equation}

\newpage 
\section{Notation}
\begin{longtable}{p{2cm}|p{12cm}}

   $b(G)$ & The number of bridges of the graph $G$. A bridge is an edge which, if removed, causes $c(G)$ to increase by one. \\
   $B(G)$ & The space of bicycles of the graph $G$. \\
   $C_{i,j}$ & Operation on $P(n)$ that contracts the $i$-th index with the $j$-th index. \\
   $c(G)$ & The number of connected components of the graph $G$. \\
   $\chi(G; q,n)$ & Generalised chromatic-polynomial function, where $\chi(G;0,n)$ is the regular chromatic polynomial in the variable $n$. \\
   $\mathbf{Circle}$ & The class of circle graphs. \\
   $\mathbf{Circle}^*$ & The class of circle graphs, and all $k$-lengthenings of circle graphs. \\
   $\mathrm{coll}(f)$ & The count of the number of collisions in the coloring described by $f$. \\
   $d$  & Dimension of Euclidean spacetime, analytically continued from integer values to arbitrary complex values.  \\
   $\partial, \delta$ & Boundary and coboundary maps used to define the bicycle number. \\ 
   $\FP$ & Function problems which can be computed in polynomial time. \\
   $G \star G'$  & For two graphs $G,G'$ with subsets of vertices that are labelled the same, $G \star G'$ is the graph obtained by gluing the vertex subsets together. \\
   $\mathrm{Gr}(x)$ & For $x \in P(n)$, the graph formed by placing the elements of $x$ around a circle, connecting contracted elements by straight chords, and taking the intersection graph. \\
   $\mathrm{Gr}'(x)$ & For $x \in P(n)$, the graph formed on $2n$ labelled vertices $\{1,\dots,2n\}$ such that $i$ and $j$ are connected by an edge if and only if $i \neq j$ and $x_i = x_j$. \\
   $K_n$ & Complete graph on $n$ vertices. \\
   $K_{m,n}$ & Complete bipartite graph (an edge connecting every node in a set of $m$ vertices to every node in a set of $n$ vertices).\\
   $l(G)$ & The number of loops of the graph $G$. A loop is an edge which connects a vertex to itself. \\
   $\NP$ & The class of decision problems for which a verification Turing machine can verify a given problem has a solution in polynomial time given a polynomial length certificate. \\
   $\#\P$ & The class of function problems corresponding to counting the number of certificates of a given $\NP$ decision problem. \\
   $\P$ & The class of decision problems solvable in polynomial time. \\
   $\pi_X(\mathcal{C})$ & Function problem corresponding to computing the Tutte polynomial on some subset $X \in \mathbb{R}^2$ for the family of graphs described by $\mathcal{C}$. \\
   $P(n)$ & Set of all permutations of tuples of length $2n$ with symbols from $\Sigma$, such that each symbol appears exactly twice. \\
   $Q(n)$ & Set of perfect matchings on the complete graph $K_{2n}$. \\
   $\Sigma$ & Countable alphabet of symbols, usually labelled as $\Sigma = \{\mu_1,\mu_2,\dots\}$. \\
   $S_{i,j}$ & Operation on $P(n)$ that swaps the $i$-th element in the tuple with the $j$-th element in the tuple. \\
   $S_n$ & Symmetric group of permutations on $n$ elements. \\
   $\mathrm{src}(e),\mathrm{snk}(e)$ & The source and sink/destination of an oriented edge $e$. \\ 
   $T(G;x,y)$ & The Tutte polynomial of the graph $G$, in the variables $x$ and $y$. \\
   $\mathrm{Tr}_d$ & The dimensionally regulated Dirac trace in $d$-dimensions. \\
   $\mathbb{Z}[x_1,\dots,x_n]$ & Ring of polynomials in the variables $x_1,\dots,x_n$ with integer coefficients. \\
\end{longtable}

\end{appendix}

\bibliographystyle{unsrtnat}
\bibliography{refs} 

\end{document}